\documentclass[ejs]{imsart}

\RequirePackage{amsthm,amsmath,amsfonts,amssymb}
\RequirePackage[numbers]{natbib}
\RequirePackage[colorlinks,citecolor=blue,urlcolor=blue]{hyperref}
\usepackage{graphicx}

\usepackage{mathtools}
\usepackage{psfrag,epsf}
\usepackage{enumerate}
\usepackage{centernot}
\usepackage[bottom]{footmisc}
\usepackage{indentfirst}
\usepackage{endnotes}
\usepackage{rotating}
\usepackage{booktabs}
\usepackage{algorithm}
\usepackage{array}
\usepackage{multirow}
\usepackage{algpseudocode}
\usepackage{enumerate}
\usepackage{bbm}
\usepackage{dsfont}
\usepackage{verbatim}
\usepackage{subfig}
\usepackage{xr}
\usepackage{relsize}
\usepackage{color}
\usepackage[top=1.1in, bottom=1.1in, left=1.3in, right=1.3in]{geometry}
\usepackage{tabularx}
\newcolumntype{L}{>{\raggedright\arraybackslash}X}

\newcommand{\be} {\begin{eqnarray*}}
	\newcommand{\ee} {\end{eqnarray*}}

\newcommand{\dCov}{{\textrm{dCov}}}
\newcommand{\CdCov}{{\textrm{CdCov}}}

\makeatletter
\newcases{nocases}
{\quad}
{$\m@th\displaystyle{##}$\hfil}
{$\m@th\displaystyle{##}$\hfil}
{.}{.}
\makeatother

\newcommand{\field}[1]{\mathbb{#1}}

\newcommand{\E}{\field{E}}
\newcommand{\tran}{^{\top\kern-\scriptspace}}

\def\qed{\hfill$\diamondsuit$}

\newcommand{\Rmnum}[1]{\expandafter\romannumeral #1}

\def\bb{\mathbb}
\def\cal{\mathcal}
\def\dis{\displaystyle}

\newcommand{\bigCI}{\mathrel{\text{\scalebox{1.07}{$\perp\mkern-10mu\perp$}}}}
\newcommand{\nbigCI}{\centernot{\bigCI}}
\DeclareMathOperator{\tr}{tr}

\theoremstyle{plain}

\newtheorem{theorem}{Theorem}[section]
\newtheorem{remark}{Remark}[section]

\newtheorem{lemma}{Lemma}[section]
\newtheorem{example}{Example}[section]
\theoremstyle{remark}
\newtheorem{definition}[theorem]{Definition}
\startlocaldefs
\numberwithin{equation}{section}
\theoremstyle{plain}

\endlocaldefs

\begin{document}
	
	\begin{frontmatter}
		\title{Nonparametric causal structure learning \\in high dimensions} %\thanksref{T1}}
		\runtitle{Nonparametric causal structure learning in high dimensions}
		%\thankstext{T1}{Footnote to the title with the `thankstext' command.}
		
		\begin{aug}
			\author{\fnms{Shubhadeep} \snm{Chakraborty}\ead[label=e1]{deep20@uw.edu}}
			%\and
			%\author{\fnms{Xianyang} \snm{Zhang}\thanksref{t3}\ead[label=e2]{second@somewhere.com}}
			
			\address{Department of Biostatistics, University of Washington\\
				\printead{e1}}
			
			\author{\fnms{Ali} \snm{Shojaie}
				\ead[label=e3]{ashojaie@uw.edu}}
			
			\address{Department of Biostatistics, University of Washington\\
				%usually few lines long\\
				%usually few lines long\\
				\printead{e3}\\
			}
			
			%\thankstext{t1}{Some comment}
			%\thankstext{t2}{First supporter of the project}
			%\thankstext{t3}{Second supporter of the project}
			\runauthor{S. Chakraborty and A. Shojaie}
			
			%\affiliation{Some University and Another University}
			
		\end{aug}
		
		\begin{abstract}
			The PC and FCI algorithms are popular constraint-based methods for learning the structure of directed acyclic graphs (DAGs) in the absence and presence of latent and selection variables, respectively. These algorithms (and their order-independent variants, PC-stable and FCI-stable) have been shown to be consistent for learning sparse high-dimensional DAGs based on partial correlations. However, inferring conditional independences from partial correlations is valid if the data are jointly Gaussian or generated from a linear structural equation model --- an assumption that may be violated in many applications. To broaden the scope of high-dimensional causal structure learning, we propose nonparametric variants of the PC-stable and FCI-stable algorithms that employ the conditional distance covariance (CdCov) to test for conditional independence relationships. 
			% Our proposed algorithms account for non-linear and non-monotone conditional dependences among the random variables and generalize the PC and FCI algorithms to arbitrary \as{correct?} distributions over DAGs.
			As the key theoretical contribution, we prove that the high-dimensional consistency of the PC-stable and FCI-stable algorithms carry over to general distributions over DAGs when we implement CdCov-based nonparametric tests for conditional independence. Numerical studies demonstrate that our proposed algorithms perform nearly as good as the PC-stable and FCI-stable for Gaussian distributions, and offer advantages in non-Gaussian graphical models.
		\end{abstract}
		
		%\begin{keyword}[class=MSC]
		%\kwd[Primary ]{60K35}
		%\kwd{60K35}
		%\kwd[; secondary ]{60K35}
		%\end{keyword}
		
		%\tableofcontents
	\end{frontmatter}

	\section{Introduction}\label{sec:intro}
	
	Directed acyclic graphs (DAGs) are commonly used to represent causal relationships among random variables (Lauritzen, 1996; Spirtes et al., 2000; Maathuis et al., 2019). %Methods for estimating DAGs from observational data can be broadly categorized as constraint-based, score-based and hybrid methods. 
	The PC algorithm (Spirtes et al., 2000) is the most popular constraint-based method for learning DAGs from observational data under the assumption of causal sufficiency, i.e., when there are no unmeasured common causes and no selection variables. It first estimates the skeleton of a DAG by recursively performing a sequence of conditional independence tests, and then uses the information from the conditional independence relations to partially orient the edges, resulting in a completed partially directed acyclic graph (CPDAG). In Section \ref{sec:overview}, we provide a review of these and other notions commonly used in the graphical modeling literature that are relevant to our work. Also we refer to estimating the CPDAG as structure learning of the underlying DAG throughout the rest of the paper. %Score-based methods aim to optimize a scoring criterion over the space of possible DAGs or CPDAGs, typically via a greedy search procedure (Chickering\,(2002a, 2002b), Nandy et al.\,(2018), etc.). And hybrid methods combine constraint and score-based approaches (Tsamardinos et al.\,(2006), Schmidt et al.\,(2007), etc.).
	
	Observational studies often involve latent and selection variables, which complicate the causal structure learning problem. Ignoring such unmeasured variables can make the causal inference based on the PC algorithm erroneous; see, e.g., Section 1.2 in Colombo et al.\,(2012) for some illustrations. The Fast Causal Inference (FCI) algorithm and its variants (Spirtes et al., 2000; Spirtes et al., 2001; Zhang, 2008; Colombo et al., 2012) utilize similar strategies as the PC algorithm to learn the DAG structure in the presence of latent and selection variables. 
	
	Both PC and FCI algorithms adopt a hierarchical search strategy --- they %starting from a complete undirected graph, recursively removing edges via 
	recursively perform conditional independence tests given subsets of increasingly larger cardinalities in some appropriate search pool.
	The PC algorithm is usually order-dependent, in the sense that its output depends on the order in which pairs of adjacent vertices and subsets of their adjacency sets are considered. The FCI algorithm suffers from a similar limitation. To overcome this limitation, Colombo and Maathuis\,(2014) proposed two variants of the PC and FCI algorithms, namely the PC-stable and FCI-stable algorithms, that resolve the order dependence at different stages of the algorithms.
	
	In general, testing for conditional independence is a problem of central importance in the causal structure learning. The literature on the PC and FCI algorithms predominantly uses partial correlations to infer conditional independence relations. It is well-known that the characterization of conditional independence by partial correlations, or in other words, equivalence between conditional independence and zero partial correlations only holds for multivariate normal random variables. Therefore, the high-dimensional consistency results for the PC and FCI algorithms (Kalisch and B\"uhlmann, 2007; Colombo et al., 2012) are limited to Gaussian graphical models, where the nodes correspond to random variables with a joint Gaussian distribution. Although the Gaussian graphical model is the standard parametric model for continuous data, it may not hold in many real data applications. 
	% In particular, Gaussian graphical models are not appropriate when the observations are categorical, discrete, have heavy-tail distributions, or their support is a subset of the real line. 
	Although this limitation can be somewhat relaxed by considering linear structural equation models (SEMs) with general noise distributions (Loh and B\"uhlmann, 2014), linear SEMs and joint Gaussianity are essentially equivalent (Voorman et al., 2014). Moreover, neither approach is appropriate when the observations are categorical, discrete, have heavy-tail distributions, or their support is a subset of the real line. In Section \ref{sec:num_real}, for example, we present a real application where all the observed variables are categorical, and therefore far from being Gaussian. As an improvement, Harris and Drton\,(2013) used rank-based partial correlations to test for conditional independence relations, showing that the high-dimensional consistency of the PC algorithm holds for a broader class of Gaussian copula models. 
	Some nonparametric versions of the PC algorithm have been also proposed in the literature via kernel-based tests for conditional independence (Sun et al., 2007; Zhang et al., 2018); however, they lack theoretical justifications of the correctness of the algorithms, and are not studied in high dimensions. 
	
	%A lot of research has been done on distance, kernel, graph, rank, information theory and copula-based approaches for nonparametric independence testing.
	%; see for example Gretton et al.\,(2005, 2007), Bergsma and Dassios\,(2014), Weihs et al.\,(2018), P\'{o}czos et al.\,(2012), Roy et al.\,(2018), Berrett and Samworth\,(2019), Heller et al.\,(2013), Biswas et al.\,(2016), Sarkar et al.\,(2018), Chatterjee\,(2019), just to mention a few. 
	
	This work aims to broaden the applicability of the PC-stable and FCI-stable algorithms to general distributions by employing a nonparametric test for conditional independence relationships. To this end, we utilize recent developments on dependence metrics that quantify non-linear and non-monotone dependence between multivariate random variables. More specifically, our work builds on the idea of distance covariance (dCov) proposed by Sz\'{e}kely et al.\,(2007) and its extension to conditional distance covariance (CdCov) by Wang et al.\,(2015) as a nonparametric measure of non-linear and non-monotone conditional independence between two random vectors of arbitrary dimensions given a third.
	% we build on the a seminal work, Sz\'{e}kely et al.\,(2007) introduced the idea of distance covariance (dCov) as a distance-based dependence metric that completely characterizes dependence between two random vectors of arbitrary dimensions. This has been extended by Wang et al.\,(2015) to introduce conditional distance covariance (CdCov) as a nonparametric measure of non-linear and non-monotone conditional independence between two random vectors of arbitrary dimensions given a third. 
	Utilizing this flexibility, we use the conditional distance covariance (CdCov) to test for conditional independence relationships in the sample versions of the PC-stable and FCI-stable algorithms. 
	The resulting algorithms --- which, for distinction, are named `nonPC' and `nonFCI' --- facilitate causal structure learning from general distributions over DAGs and are shown to be consistent in sparse high-dimensional settings. 
	Our consistency results only require mild moment and tail conditions on the set of variables, without requiring any strict distributional assumptions. To our knowledge, the proposed generalizations of PC/PC-stable or the FCI/FCI-stable algorithms provide the first general nonparametric framework for causal structure learning with theoretical guarantees in high dimensions. 
	
	% Using CdCov, we propose the `nonPC' and `nonFCI' algorithms, which employ conditional distance covariance to test for conditional independence relationships in the sample versions of the PC-stable and FCI-stable algorithms, respectively, instead of partial correlations. The consistency of our proposed algorithms are established in sparse high-dimensional settings. This essentially broadens the scope or applicability of the PC/PC-stable and FCI/FCI-stable algorithms to general distributions supported by the DAG, and enables taking into account non-linear and non-monotone conditional dependence among the random variables, something that partial correlations fail to capture. Our consistency results only require some mild moment and tail conditions on the set of variables, without requiring any strict distributional assumptions. We are not aware of any previous work in the graphical modeling literature that extends the PC/PC-stable or the FCI/FCI-stable algorithms to a general nonparametric framework along with theoretical guarantees on their consistency properties in high dimensions. 
	
	The rest of the paper is organized as follows. In Section~\ref{sec:overview}, we review the relevant background, including preliminaries on graphical modeling (Section~\ref{subsec:pc}) and a brief overview of dCov and CdCov (Section~\ref{subsec:dcov}). The nonparametric version of the PC-stable algorithm is presented in Section~\ref{subsec:methods_nonPC}. As a key contribution of the paper, we establish that the algorithm consistently estimates the skeleton and the equivalence class of the underlying sparse high-dimensional DAG in a general nonparametric framework. We then present the nonparametric version of the FCI-stable algorithm in Section~\ref{subsec:methods_nonFCI} and establish its consistency in sparse high-dimensional settings. 
	As the FCI involves the adjacency search of the PC algorithm, any improvement on the PC/PC-stable directly carries over to the FCI/FCI-stable as well. 
	%The theoretical guarantees of the nonFCI essentially rests upon that of the nonPC in high dimensions. 
	In Section \ref{sec:num}, we compare the performances of our algorithms with the PC-stable and FCI-stable using both simulated datasets (involving both Gaussian and non-Gaussian examples), as well as a real dataset. These numerical studies clearly demonstrate that nonPC and nonFCI algorithms are comparable with PC-stable and FCI-stable for Gaussian data and offer improvements for non-Gaussian data. 

	\section{Background}\label{sec:overview}
	
	\subsection{Preliminaries on graphical modeling}\label{subsec:pc}
	
	A graph $\cal{G} = (V,E)$ consists of a vertex set $V=\{1, \dots, p\}$ and an edge set $E \subseteq V \times V$. In a graphical model, the vertices or nodes are associated with random variables $X_a$ for $1\leq a \leq p$. Throughout, we index the nodes by the corresponding random variables. 
	We also allow the edge set $E$ of the graph $\cal{G}$ to contain (a subset of) the following six types of edges: $\rightarrow$ (\textit{directed}), $\leftrightarrow$ (\textit{bidirected}), $-$ (\textit{undirected}), $\circ \hspace{-0.06in}-\hspace{-0.06in}\circ$ (\textit{nondirected}), $\circ\hspace{-0.02in}- $ (\textit{partially undirected}) and $\circ\hspace{-0.07in}\rightarrow $ (\textit{partially directed}). The endpoints of an edge are called marks, which can be tails, arrowheads or circles. We use the symbol `$\star$' to denote an arbitrary edge mark.  A \textit{mixed graph} is a graph containing directed, bidirected and undirected edges. 
	A graph containing only directed edges $(\rightarrow)$ is called a \textit{directed graph}, one containing only undirected edges $(-)$ is called an \textit{undirected graph}, and one containing directed and undirected edges is called a \textit{partially directed graph}. 
	%Essentially directed, undirected and partially directed graphs are special cases of mixed graphs. \as{needed?}
	%\as{below and elsewhere: although it is ok to refer to nodes of the graph as $X_a$, this needs to be clarified as the nodes are indexed by $a$ and the random variables $X_a$ \textit{correspond} to nodes of the graph.} \ch{(clarified it above, highlighted in blue)}  \as{you can either replace the $X_a$ by $a$ when eg defining $\mathrm{adj}$ (this is what I often do) or say that we use $a$ and $X_a$ interchangeably --- either is ok} \ch{(I think I used the notation $\mathrm{adj}(\cal{G}, X_a)$ throughout, not $\mathrm{adj}(\cal{G}, a)$)} 
	
	The \textit{adjacency set} of a vertex $X_a$ in the graph $\cal{G} = (V,E)$, denoted $\mathrm{adj}(\cal{G}, X_a)$, is the set of all vertices in $V$ that are adjacent to $X_a$, or in other words, are connected to $X_a$ by an edge. The \textit{degree} of a vertex $X_a$, $|\mathrm{adj}(\cal{G}, X_a)|$, is defined as the number of vertices adjacent to it. A graph is \textit{complete} if all pairs of vertices in the graph are adjacent. A vertex $X_b \in \mathrm{adj}(\cal{G}, X_a)$ is called a \textit{parent} of $X_a$ if $X_b \rightarrow X_a$, a \textit{child} of $X_a$ if $X_a \rightarrow X_b$\, %, a \textit{spouse} of $X_a$ if $X_a \leftrightarrow X_b$ 
	and a \textit{neighbor} of $X_a$ if $X_a - X_b$. 
	%\sout{The corresponding set of parents, children, spouses and neighbors of $X_a$ are denoted by $\mathrm{pa}(\cal{G}, X_a)$, $\mathrm{ch}(\cal{G}, X_a)$, $\mathrm{sp}(\cal{G}, X_a)$ and $\mathrm{ne}(\cal{G}, X_a)$, respectively.} 
	%\as{you may want to say that when clear we drop $\cal{G}$ and write eg $\mathrm{ne}(X_a)$; also do you need all of these, eg children and spouses, later? finally, please replace $ne$ etc by $\mathrm{ne}$} \ch{(I striked out the above line as I think we only used the notions of parents, children and neighbors, but never used these notations anywhere.)} \asr{thanks. perhaps we should not define spouse either; however, the way you define neighbor seems to be limited to undirected edges here --- is that how you use it later as well, or do you use it more generally to refer to any node in the adjacency set (that's usually how it is used, but ok if we define it differently, as long as we are consistent)}
	The \textit{skeleton} of the graph $\cal{G}$ is the undirected graph obtained by replacing all the edges of $\cal{G}$ by undirected edges, in other words, ignoring all the edge orientations. Three vertices $\langle X_a, X_b, X_c \rangle$ are called an \textit{unshielded triple} if $X_a$ and $X_b$ are adjacent, $X_b$ and $X_c$ are adjacent, but $X_a$ and $X_c$ are not adjacent. A \textit{path} is a sequence of distinct adjacent vertices. 
	A node $X_a$ is an \textit{ancestor} of its \textit{descendent} $X_b$, if $\cal{G}$ contains a directed path $X_a \to \cdots \to X_b$. 
	A non-endpoint vertex $X_a$ on a path is called a collider on the path if both the edges preceding and succeeding it have an arrowhead at $X_a$, or in other words, the path contains $\star\hspace{-0.05in}\rightarrow X_a \leftarrow\hspace{-0.05in}\star$. An unshielded triple $\langle X_a, X_b, X_c \rangle$ is called a \textit{v-structure} if $X_b$ is a collider on the path $\langle X_a, X_b, X_c \rangle$.
	
	%A \textit{directed path} is a sequence of distinct adjacent vertices along directed edges that follows the direction of the arrowheads. A \textit{directed cycle} is formed by a directed path from $X_a$ to $X_b$ together with the edge $X_b \rightarrow X_a$. 
	
	A \textit{cycle} occurs in a graph when there is a path from $X_a$ to $X_b$, and $X_a$ and $X_b$ are adjacent. A directed path from $X_a$ to $X_b$ forms a \textit{directed cycle} together with the edge $X_b \rightarrow X_a$, and it forms an almost directed cycle together with the edge $X_b \leftrightarrow X_a$. Three vertices that form a cycle are called a {\it triangle}. A \textit{directed acyclic graph} (DAG) is a directed graph that does not contain any cycle. A DAG entails conditional independence relationships via a graphical criterion called \textit{d-separation} (Section 1.2.3 in Pearl, 2000). Two vertices $X_a$ and $X_b$ that are not adjacent in a DAG $\cal{G}$ are d-separated in $\cal{G}$ by a subset $X_S \subseteq V \backslash \{X_a, X_b\}$. A probability distribution $P$ on ${\bb R}^p$ is said to be \textit{faithful} with respect to the DAG $\cal{G}$ if the conditional independence relationships in $P$ can be inferred from $\cal{G}$ using d-separation and vice versa, in other words, $X_a \bigCI X_b \vert X_S$ if and only if $X_a$ and $X_b$ are d-separated in $\cal{G}$ by $X_S$.
	
	A graph that is both (partially) directed and acyclic, is called a \textit{partially directed acyclic graph (PDAG)}. DAGs that encode the same set of conditional independence relations, form a Markov equivalence class (Verma and Pearl, 1990). Two DAGs belong to the same Markov equivalence class if and only if they have the same skeleton and the same v-structures. A Markov equivalence class of DAGs can be uniquely represented by a \textit{completed partially directed acyclic graph (CPDAG)}, which is a PDAG that satisfies the following : i) $X_a \rightarrow X_b$ in the CPDAG if $X_a \rightarrow X_b$ in every DAG in the Markov equivalence class, and, ii) $X_a-X_b$ in the CPDAG if the Markov equivalence class contains a DAG in which $X_a \rightarrow X_b$ as well as a DAG in which $X_a \leftarrow X_b$.
	
	Estimation of the CPDAG by the PC algorithm involves two steps: 1) estimation of the skeleton and separating sets (also called the adjacency search step) via recursive conditional independence tests; and, 2) partial orientation of edges; see Algorithms 1 and 2 in Kalisch and B\"uhlmann\,(2007) for details.%, which popularized the algorithm in sparse high-dimensional settings. Under the strong faithfulness assumption, they proved that the PC algorithm consistently estimates the skeleton and the equivalence class of an underlying sparse DAG where the number of nodes is allowed to grow at an arbitrary polynomial rate of the sample size. 
	
	In presence of latent and selection variables, one needs a generalization of a DAG, called a \textit{maximal ancestral graph} (MAG). A mixed graph is called an \textit{ancestral graph} if it contains no directed or almost directed cycles and no subgraph of the type $X_a - X_b \leftarrow\hspace{-0.05in}\star\, X_c$. DAGs form a subset of ancestral graphs. A MAG is an ancestral graph in which every missing edge corresponds to a conditional independence relationship via the m-separation criterion (Richardson and Spirtes, 2002), a generalization of the notion of d-separation. Multiple MAGs may represent the same set of conditional independence relations. Such MAGs form a Markov equivalence class which can be represented by a \textit{partial ancestral graph} (PAG) (Ali et al, 2009); see Richardson and Spirtes\,(2002) for additional details. 
	
	Under the faithfulness assumption, the Markov equivalence class of a DAG with latent and selection variables can be learned %from the conditional independence relationships among the observed variables 
	using the FCI algorithm (e.g., Algorithm~3.1 in Colombo et al., 2012), which is a modification of the PC algorithm. %We refer the reader to Algorithm 3.1 in Colombo et al.\,(2012) for a compact description of the FCI algorithm. %FCI first runs the PC algorithm to obtain a preliminary skeleton, which is a superset of the PAG skeleton. Based on that, it then computes certain sets called `Possible D-SEP' sets in the next step, and estimates the PAG skeleton using those. The output of the FCI algorithm is the estimated PAG.
	The FCI algorithm first employs the adjacency search of the PC algorithm, and then performs additional conditional independence tests because of the presence of latent variables followed by partial orientation of the edges, resulting in an estimated PAG. To estimate the skeletons (of the DAG and the PAG, respectively), both the PC and the FCI algorithms adopt a hierarchical search strategy that starts with a complete undirected graph and recursively removes edges via conditional independence tests given subsets of increasingly larger cardinalities in some appropriate search pool. Colombo and Maathuis\,(2014) proposed order-independent variants of the PC and the FCI algorithms, namely the PC-stable and FCI-stable algorithms. %and we would refer the reader to Section 4 in that paper for detailed descriptions and illustrations.
	To make the paper self-contained, we provide the pseudocodes of the oracle versions of the PC-stable and FCI-stable algorithms in Appendix~\ref{sec:preliminaries}. %\as{depending on the journal this may be the appendix etc} \ch{(okay)}

	%\chakraborty{\hspace{-0.1in}(Comment : Should we add definitions of MAGs and PAGs over here for the sake of completeness? I feel that then we need to first illustrate the idea of ancestral graphs, then m-separation criterion, MAGs, etc., which will make this section VERY long. Would love to hear your thoughts and suggestions !)}
	
	\subsection{Distance covariance and conditional distance covariance}\label{subsec:dcov}
	
	%\emph{Notations}. 
	We start by describing the notation used throughout the paper. We denote by $\Vert\cdot\Vert_p$ the Euclidean norm of $\mathbb{R}^p$ and use $\Vert\cdot\Vert$ when the dimension is clear from the context. 
	% (we shall use it interchangeably with $\Vert\cdot\Vert$ when there is no confusion). 
	We use $X \bigCI Y$ to denote the independence of $X$ and $Y$ and use $\E_U$ to denote expectation with respect to the probability distribution of the random variable $U$. For any set $S$, we denote its cardinality by $|S|$. 
	
	We use the usual asymptotic notation, `$O$' and `$o$', as well as their probabilistic counterparts, $O_p$ and $o_p$, which denote stochastic boundedness and convergence in probability, respectively.
	%stand for the usual notations in mathematics : `is of the same order as' and `is ultimately smaller than'. 
	For two sequences of real numbers $\{a_n\}_{n=1}^{\infty}$ and $\{b_n\}_{n=1}^{\infty}$, $a_n \asymp b_n$ if and only if $a_n/b_n = O(1)$ and $b_n/a_n = O(1)$ as $n \to \infty$. We use the symbol $``a \lesssim b"$ to indicate that $a \leq C \,b$\, for some constant $C>0$. 
	%We utilize the order in probability notations such as stochastic boundedness $O_p$ (big O in probability) and convergence in probability $o_p$ (small o in probability). %Denote ${\bf 1}_n = (1,\dots, 1) \in \bb{R}^n$.d 
	For a matrix $A = (a_{kl})_{k,l=1}^n \in \bb{R}^{n \times n}$, we denote its
	determinant by $|A|$ and define its\, $\cal{U}$-centered version $\tilde{A} = (\tilde{a}_{kl})_{k,l=1}^n$ as 
	\begin{align}\label{U centering def}
		\tilde{a}_{kl} = \begin{cases}
			a_{kl}\, -\,\mathlarger{ \frac{1}{n-2}} \dis\sum_{j=1}^n a_{kj} \,-\, \frac{1}{n-2} \dis\sum_{i=1}^n a_{il}\, +\, \frac{1}{(n-1)(n-2)} \dis\sum_{i,j=1}^n a_{ij}, \; & k \neq l,\\
			0, & k=l, \end{cases}
	\end{align}
	for $k, l = 1, \dots, \, n$. 
	% For any matrix $A$, denote its determinant by $|A|$. 
	Finally, we denote the integer part of $a\in\mathbb{R}$ by $\lfloor a \rfloor$.
	
	\begin{comment}
	Denote ${\bf 1}_n = (1,\dots, 1) \in \bb{R}^n$. For a matrix $A = (a_{kl})_{k,l=1}^n \in \bb{R}^{n \times n}$, define its\, $\cal{U}$-centered version $\tilde{A} = (\tilde{a}_{kl}) \in \bb{R}^{n \times n}$ as follows
	\begin{align}\label{U centering def}
	\tilde{a}_{kl} = \begin{cases}
	a_{kl}\, -\,\mathlarger{ \frac{1}{n-2}} \dis\sum_{j=1}^n a_{kj} \,-\, \frac{1}{n-2} \dis\sum_{i=1}^n a_{il}\, +\, \frac{1}{(n-1)(n-2)} \dis\sum_{i,j=1}^n a_{ij}, \; & k \neq l,\\
	0, & k=l, \end{cases}
	\end{align}
	for $k, l = 1, \dots, \, n$. 
	\end{comment}

	Sz\'ekely et al.\,(2007), in their seminal paper, introduced the notion of distance covariance (dCov, henceforth) to quantify non-linear and non-monotone dependence between two random vectors of arbitrary dimensions. Consider two random vectors $X\in \bb{R}^p$ and $Y \in \bb{R}^q$ with $\E \Vert X \Vert_p < \infty$ and $\E \Vert Y \Vert_q < \infty$. The distance covariance between $X$ and $Y$ is defined as the positive square root of
	\begin{align*}
		\dCov^2(X,Y)=\frac{1}{c_{p}c_{q}}\int_{{\mathbb
				R}^{p+q}}\frac{|f_{X,Y}(t,s)-f_X(t)f_Y(s)|^2}{\Vert t \Vert_p^{1+p}\,\Vert s \Vert_q^{1+q}}dtds\,,
	\end{align*}
	where $f_X$, $f_Y$ and $f_{X,Y}$ are the individual and joint
	characteristic functions of $X$ and $Y$ respectively, and $c_{p}=\pi^{(1+p)/2}/\,\Gamma((1+p)/2)$
	is a constant with $\Gamma(\cdot)$ being the complete gamma function.
	
	The key feature of dCov is that it completely characterizes the independence between two random vectors, or in other words $\dCov(X,Y)=0$ if and
	only if $X \bigCI Y$. According to Remark~3 in Sz\'ekely et al.\,(2007), dCov can be equivalently expressed as
	\begin{align*}
		\dCov^2(X,Y) \;=\; \E\,\Vert X-X'\Vert_p \Vert Y-Y'\Vert_q \,+\, \E\,\Vert X-X'\Vert_p\,\E\,\Vert Y-Y'\Vert_q \,-\, 2 \,\E\,\Vert X-X'\Vert_p \Vert Y-Y''\Vert_q\,. 
	\end{align*}
	
	\begin{comment}
	For an observed random sample $(X_i,Y_i)^{n}_{i=1}$ from the joint distribution of $X$ and $Y$, define the distance matrices $A=\big(A_{kl}\big)_{k,l=1}^n$ and $B=\big(B_{kl}\big)_{k,l=1}^n \in \mathbb{R}^{n\times n}$ where $A_{kl} = \Vert X_k - X_l \Vert_p$ and $B_{kl} = \Vert Y_k - Y_l \Vert_q$. Following the $\mathcal{U}$-centering idea in Sz\'ekely and Rizzo\,(2014) and Lemma A.1 in Park et al.\,(2015), Lemma 2.1 in Yao et al.\,(2018)  shows that an unbiased U-statistic type estimator of the squared dCov can be expressed as
	\begin{align*}
	dCov^2_n(X,Y)\; :&=\;  \frac{1}{n(n-3)} \left[\tr (\tilde{A} \tilde{B}) \,+\, \frac{{\bf 1}_n^T \tilde{A} {\bf 1}_n\, {\bf 1}_n^T \tilde{B} {\bf 1}_n}{(n-1)(n-2)} \,-\, \frac{2}{n-2}{\bf 1}_n^T \tilde{A}  \tilde{B} {\bf 1}_n  \right]\,.
	\end{align*}
	
	It has been pointed out in the literature that $dCov^2_n(X,Y)$ has a degeneracy of order 1, and therefore when scaled by $n$, converges weakly to a non-pivotal distribution with infinitely many nuisance parameters (see, for example, Lemma 4.10 in Huang and Huo\,(2017)). Therefore, it has been a common practice in the literature to use resampling based procedures for dCov based tests for independence. Given the quantification of non-linear and non-monotone dependence by the population quantity and the striking simplicity of the sample version, dCov has been widely studied, extended and analyzed in the literature over the last one decade.
	\end{comment}

	This alternate expression comes handy in constructing V or U-statistic type estimators for the quantity. For an observed random sample $(X_i,Y_i)^{n}_{i=1}$ from the joint distribution of $X$ and $Y$, define the distance matrices $d^X=\big(d^X_{ij}\big)_{i,j=1}^n$ and $d^Y=\big(d^Y_{ij}\big)_{i,j=1}^n \in \mathbb{R}^{n\times n}$ where $d^X_{ij}:= \Vert X_i - X_j \Vert_p$ and $d^Y_{ij}:= \Vert Y_i - Y_j \Vert_q$. Following the $\mathcal{U}$-centering idea in Sz\'ekely and Rizzo\,(2014), an unbiased U-statistic type estimator of $\dCov^2(X,Y)$ can be expressed as
	\begin{align}\label{estimator of dCov}
		\dCov^2_n(X,Y)\; &:=\;  (\tilde{d}^{\,X} \cdot\, \tilde{d}^{\,Y})\;:=\; \frac{1}{n(n-3)} \dis \sum_{i\neq j} \tilde{d}^{\,X}_{ij} \tilde{d}^{\,Y}_{ij} \; ,
	\end{align}
	where $\tilde{d}^{\,X} = (\tilde{d}^{\,X}_{ij})_{i,j=1}^n$\, and \, $\tilde{d}^{\,Y} = (\tilde{d}^{\,Y}_{ij})_{i,j=1}^n$ are the \,$\cal{U}$-centered versions of the matrices $\tilde{d}^{\,X}$ and $\tilde{d}^{\,Y}$, respectively, as defined in (\ref{U centering def}). 
	
	Wang et al.\,(2015) recenlty generalized the notion of dCov and  introduced the conditional distance covariance (CdCov, henceforth) as a measure of conditional dependence between two random vectors of arbitrary dimensions given a third. CdCov essentially replaces the characteristic functions used in the definition of dCov by conditional characteristic functions.  
	Consider a third random vector $Z\in \bb{R}^r$ with $\E (\Vert X \Vert_p + \Vert Y \Vert_q \mid Z) < \infty$. Denote by $f_{X,Y|Z}$ the conditional joint characteristic function of $X$ and $Y$ given $Z$, and by $f_{X|Z}$ and $f_{Y|Z}$ the conditional marginal characteristic functions of $X$ and $Y$ given $Z$, respectively. Then CdCov between $X$ and $Y$ given $Z$ is defined as the positive square root of 
	\begin{align*}
		\CdCov^2(X,Y|Z)=\frac{1}{c_{p}c_{q}}\int_{{\mathbb
				R}^{p+q}}\frac{|f_{X,Y|Z}(t,s)-f_{X|Z}(t)f_{Y|Z}(s)|^2}{\Vert t \Vert_p^{1+p}\,\Vert s \Vert_q^{1+q}}dtds.
	\end{align*}
	The key feature of CdCov is that $\CdCov\,(X,Y|Z)=0$ almost surely if and
	only if $X \bigCI Y | Z$, which is quite straightforward to see from the definition.

	Similar to dCov, an equivalent alternative expression can be established for CdCov that avoids complicated integrations involving conditional characteristic functions. Let $W_i = (X_i, Y_i,\\ Z_i)_{i=1}^n$ be an i.i.d. sample from the joint distribution of $W:=(X,Y,Z)$. Define $d_{ijkl} := \big(d^X_{ij} + d^X_{kl} - d^X_{ik} - d^X_{jl}\big)\, \big(d^Y_{ij} + d^Y_{kl} - d^Y_{ik} - d^Y_{jl}\big)$, which is not symmetric with respect to $\{i,j,k,l\}$, and therefore necessitates defining the following symmetric form: $d^S_{ijkl} := d_{ijkl} \,+\,d_{ijlk}\,+\,d_{ilkj}$.
	Lemma 1 in Wang et al.\,(2015) establishes an equivalent representation of $\CdCov^2(X,Y|Z=z)$ as
	\begin{align}\label{cdcov equiv}
		\CdCov^2(X,Y|Z=z) \;=\; \frac{1}{12}\, \E\,\big[ d^S_{1234}\,\vert\, Z_1=z, Z_2=z, Z_3=z, Z_4=z \big]\,.
	\end{align}
	
	\begin{remark}\label{remark1}
		In a recent work, Sheng and Sriperumbudur\,(2019) explore the connection between conditional independence measures induced by distances on a metric space and reproducing kernels associated with a reproducing kernel Hilbert space (RKHS). They generalize CdCov to arbitrary metric spaces of negative type --- termed generalized CdCov (gCdCov) --- and develop a kernel-based measure of conditional independence, namely the  Hilbert-Schmidt conditional independence criterion (HSCIC). Theorem~1 in their paper establishes an equivalence between gCdCov and HSCIC, or in other words, between distance and kernel-based measures of conditional independence. 

	\end{remark}

	For $w \in {\bb R}^r$, let $K_H(w):=|H|^{-1}\,K(H^{-1} w)$ be a kernel function where $H$ is the diagonal matrix $\textrm{diag}(h,\dots,h)$ determined by a bandwidth parameter $h$. $K_H$ is typically considered to be the Gaussian kernel $K_H(w) = (2\pi)^{-\frac{r}{2}}\,|H|^{-1}\, \exp\big(-\frac{1}{2} w^T H^{-2} w \big)$, where $w \in {\bb R}^r$.
	
	Let $K_{iu} := K_H(Z_i-Z_u) = |H|^{-1}\,K(H^{-1} (Z_i-Z_u))$ and $K_i(Z):=K_H(Z-Z_i)$ for $1\leq i,u \leq n$. Then by virtue of the equivalent representation of CdCov in (\ref{cdcov equiv}), a V-statistic type estimator of $\CdCov ^2(X,Y|Z)$ can be constructed as 
	\begin{align}\label{cdcov est}
		\CdCov^2_n(X,Y|Z) \;:=\; \dis \sum_{i,j,k,l} \frac{K_i(Z)\, K_j(Z)\, K_k(Z)\, K_l(Z)}{12\,\big(\sum_{i=1}^n K_i(Z)\big)^4}\, d^S_{ijkl}\,.
	\end{align}
	Under certain regularity conditions, Theorem~4 in Wang et al.\,(2015) shows that conditioned on $Z$, $\CdCov^2_n(X,Y|Z) \overset{P}{\longrightarrow} \CdCov^2(X,Y|Z)$ as $n\to \infty$.

	\section{Methodology and Theory}\label{sec:methods}
	%\subsection{The nonPC algorithm}\label{subsec:methods_nonPC}
	\subsection{The Nonparametric PC Algorithm in High Dimensions}\label{subsec:methods_nonPC}
	
	To get a measure of conditional independence between $X$ and $Y$ given $Z$ that is free of $Z$, we define 
	\begin{align}\label{rho0*}
		\rho_0^*\,(X, Y | Z) \;:=\; \E\,\big[ \CdCov^2_n(X,Y|Z) \big]\,.
	\end{align}
	Clearly $\rho_0^*\,(X, Y | Z)=0$ if and only if $X \bigCI Y \,|\, Z$. Consider a plug-in estimate of $\rho_0^*\,(X, Y | Z)$ as 
	\begin{align}\label{rho0*_est}
		\begin{split}
			&\widehat{\rho}^{\;*}(X, Y | Z) \;:=\; \frac{1}{n} \dis \sum_{u=1}^n \CdCov^2_n(X,Y|Z_u)\, \;=\; \frac{1}{n} \dis \sum_{u=1}^n \Delta_{i,j,k,l;u}\,,\\
			\textrm{where}\qquad &\Delta_{i,j,k,l;u} \;:=\; \dis \sum_{i,j,k,l} \frac{K_{iu}\, K_{ju}\, K_{ku}\, K_{lu}}{12\,\big(\sum_{i=1}^n K_{iu}\big)^4}\, d^S_{ijkl}\,.
		\end{split}
	\end{align}
	We reject $H_0 : X \bigCI Y | Z$\, vs\, $H_A : X \nbigCI Y | Z$ at level $\alpha \in (0,1)$ if \,$\widehat{\rho}^{\;*}(X, Y | Z) > \xi_{n,\alpha}$, where the threshold $\xi_{n,\alpha}$ is typically obtained by a local bootstrap procedure (see Section 4.3 in Wang et al., 2015). Henceforth we will often denote $\rho_0^*\,(X, Y | Z)$ simply by $\rho_0^*$ for notational simplicity, whenever there is no confusion.
	
	In view of the complete characterization of conditional independence by $\rho_0^*$, we propose testing for conditional independence relations nonparametrically in the sample version of the PC-stable algorithm based on $\rho_0^*$, rather than partial correlations. We coin the resulting algorithm the `nonPC' algorithm, to emphasize that it is a nonparametric generalization of parametric PC-stable algorithms. 
	
	The \textit{oracle version} of the first step of nonPC, or the skeleton estimation step, is exactly the same as that of the PC-stable algorithm (Algorithm~\ref{PC-stable1} in Appendix~\ref{sec:preliminaries}). The second step which extends the skeleton estimated in the first step to a CPDAG (Algorithm~\ref{PC-stable2} in Appendix~\ref{sec:preliminaries}) comprises of some purely deterministic rules for edge orientations, and is exactly the same for both the nonPC and PC-stable as well. The only difference lies in the implementation of the tests for conditional independence relationships in the \textit{sample versions} of the first step. Specifically, we replace all the conditional independence queries in the first step by tests based on $\rho_0^*\,(X, Y | Z)$. At some pre-specified significance level $\alpha$, we infer that $X_a \bigCI X_b \,|\, X_S$ when $\widehat{\rho}^{\;*}(X_a, X_b | X_S) \leq \xi_{n,\alpha}$,
	%\begin{align}\label{CI test}
	%X_a \bigCI X_b \,|\, X_S \; \iff \; \widehat{\rho}^{\;*}(X_a, X_b | X_S) \leq \xi_{n,\alpha}\,,
	%\end{align}
	where $a, b \in V$ and $S \subseteq V$, $|S|\neq\phi$. When $|S|=\phi$,  $\widehat{\rho}^{\;*}(X_a, X_b | X_S) = \dCov_n^2(X_a,X_b)$ and $\rho_0^*\,(X, Y | Z) = \dCov^2(X,Y)$. The critical value $\xi_{n,\alpha}$ in this case is obtained by a bootstrap procedure (see, e.g., Section~4 in Chakraborty and Zhang, 2019 with $d=2$).
	
	Given that the equivalence between conditional independence and zero partial correlations only holds for multivariate normal random variables, our generalization broadens the scope of applicability of causal structure learning by the PC/PC-stable algorithm to general distributions over DAGs. This nonparametric approach is thus a natural extension of Gaussian and Gaussian copula models. It enables capturing non-linear and non-monotone conditional dependence relationships among the variables, which partial correlations fail to detect. %We will refer to the PC-stable algorithm that uses the CI tests based on $\rho_0^*$ as the `nonPC' algorithm. 
	%
	%\chakraborty{Comment : I feel that overall I'm perhaps being somewhat hand wavy by not including the algorithmic structure of the nonPC/nonFCI in the paper, but just mentioning what modification we are proposing to the PC-stable/FCI-stable algorithms. Is that making the paper not look self-contained? The point is, the oracle version of our algorithms won't be any different than PC-stable/FCI-stable; the difference would be in their sample versions. Kalisch and B\"uhlmann\,(2007), Colombo and Maathuis\,(2014) everyone mentions the oracle PC/PC-stable first, and then mentions in a couple of lines what changes to make to get the sample version. Including the oracle version seems a wastage of space, when that is not anything new in the literature. Would love to hear suggestions !}
	
	Next, we establish theoretical guarantees on the correctness of the nonPC algorithm in learning the true underlying causal structure in sparse high-dimensional settings. Our consistency results only require mild moment and tail conditions on the set of variables, without making any strict distributional assumptions. 
	Denote by $m_p$ the maximum cardinality of the conditioning sets considered in the adjacency search step of the PC-stable algorithm. Clearly $m_p \leq q$, where $q := \max_{1\leq a \leq p} \,| \mathrm{adj}(\mathcal{G}, a)\,|$ is the maximum degree of the DAG $\mathcal{G}$. For a fixed pair of nodes $a, b \in V$, the conditioning sets considered in the adjacency search step are elements of $J_{a,b}^{m_p} := \{S \subseteq V \backslash \{a,b\} : |S| \leq m_p \}$.
	
	We first establish a concentration inequality that gives the rate at which the absolute difference of $\rho_0^*\,(X_a, X_b | X_S)$ and its plug-in estimate $\widehat{\rho}^{\;*}(X_a, X_b | X_S)$ decays to zero, for any fixed pair of nodes $a$ and $b \in V$ and a fixed conditioning set $S$. Towards that, we impose the following regularity conditions.
	
	\begin{itemize}
		\item[(A1)] There exists $s_0 >0$ such that for $0\leq s<s_0$, $\dis \sup_p \max_{1\leq a\leq p} \E\,\exp(s X_a^2)<\infty$. %\as{do we need the sup over p?} \ch{(I think so. In page 31, when we use Lemma 2 in the supplementary materials of Wen et al.\,(2018), I think we require condition C2 in their paper to hold, which is similar to what we put in our condition (A1).)}
		\item[(A2)] The kernel function $K(\cdot)$ is non-negative and uniformly bounded over its support. 
	\end{itemize}
	
	Condition (A1) imposes a sub-exponential tail bound on the random variables. This is a quite commonly used condition, for example, in the high-dimensional feature screening literature (see, for example, Liu et al., 2014). Condition (A2) is a mild condition on the kernel function $K(\cdot)$ that is guaranteed by many commonly used kernels,  including the Gaussian kernel. 
	Under conditions (A1) and (A2), the next result shows that the plug-in estimate $\widehat{\rho}^{\;*}(X_a, X_b | X_S)$ converges in probability to its population counterpart $\rho_0^*\,(X_a, X_b | X_S)$ exponentially fast.
	
	\begin{theorem}\label{th:conc}
		Under conditions (A1) and (A2), for any\, $\epsilon > 0$ there exist positive constants $A$, $B$ and\, $\gamma \in (0,1/4)$\, such that
		%\vspace{-0.2in}
		\begin{align*}
			\mathbb{P}\left(\vert\, \widehat{\rho}^{\;*}(X_a, X_b | X_S) - \rho_0^*\,(X_a, X_b | X_S) \,\vert > \epsilon \right) \;\leq \;  O\Big( 2\,\exp\left( - A\,n^{1-2\gamma}\,\epsilon^2  \right) \,+\, n^4\, \exp\big( - B \, n^{\gamma}\big)\Big)\,.
		\end{align*}
	\end{theorem}
	
	The proof of Theorem \ref{th:conc} is long and somewhat technical;  it is thus relegated to Section~\ref{sec:appendix}. Theorem \ref{th:conc} serves as the main building block towards establishing the consistency of the nonPC algorithm in sparse high-dimensional settings.

	For notational convenience, henceforth we denote $\rho_0^*\,(X_a, X_b | X_S)$ and $\widehat{\rho}^{\;*}(X_a, X_b | X_S)$ by $\rho_{0\,;\, a b|S}^*$ and $\widehat{\rho}^{\;*}_{a b|S}$, respectively. In Theorem \ref{th:conc_sup} below, we establish a uniform bound for the errors in inferring conditional independence relationships using the $\rho_0^*$-based test in the skeleton estimation step of the sample version of the nonPC algorithm.
	
	\begin{theorem}\label{th:conc_sup}
		Under conditions (A1) and (A2), for any\, $\epsilon > 0$ there exist positive constants $A$, $B$ and\, $\gamma \in (0,1/4)$\, such that
		\begin{align}\label{conc_sup}
			\begin{split}
				\sup_{\substack{a,b \in V \\ S \in J_{a,b}^{m_p}}}  \mathbb{P}\big(\vert\, \widehat{\rho}^{\;*}_{ab|S} - \rho_{0\,;\, ab|S}^* \,\vert > \epsilon \big) \;& \leq \; \mathbb{P}\Big(\sup_{\substack{a,b \in V \\ S \in J_{a,b}^{m_p}}} \vert\, \widehat{\rho}^{\;*}_{ab|S} - \rho_{0\,;\, ab|S}^* \,\vert > \epsilon \Big) \\
				& \leq \; O\Big(p^{m_p+2}\, \big[\, 2\,\exp\big(-A\, n^{1-2\gamma}\,\epsilon^2 \big) \,+\,  n^4 \,\exp\big(- B\, n^{\gamma}  \big) \big]\Big)\,.
			\end{split}
		\end{align}
	\end{theorem}
	
	\begin{comment}
	\begin{theorem}\label{th:conc_sup}
	Under conditions (A1) and (A2), for any\, $c > 0$\, and\, $\kappa>0$\, there exist positive constants $A$, $B$ and\, $\gamma$\, satisfying\, $0<\kappa+\gamma<1/4$, such that
	\begin{align}\label{conc_sup}
	\begin{split}
	\sup_{\substack{a,b \in V \\ S \in J_{a,b}^{m_p}}}  P\big(\vert\, \hat{\rho}^*_{ab|S} - \rho_{0\,;\, ab|S}^* \,\vert > c\,n^{-\kappa} \big) \;& \leq \; P\Big(\sup_{\substack{a,b \in V \\ S \in J_{a,b}^{m_p}}} \vert\, \hat{\rho}^*_{ab|S} - \rho_{0\,;\, ab|S}^* \,\vert > c\,n^{-\kappa} \Big) \\
	& \leq \; O\Big(p^{m_p+2}\, \big[\, 2\,\exp\big(-A\, n^{1-2(\kappa+\gamma)} \big) \,+\,  n^4 \,\exp\big(- B\, n^{\gamma}  \big) \big]\Big)\,.
	\end{split}
	\end{align}
	\end{theorem}
	\end{comment}

	Next, we turn to proving the consistency of the nonPC algorithm in the high-dimensional setting where the dimension $p$ can be much larger than the sample size $n$, but the DAG is considered to be sparse. We impose the following regularity conditions, which are similar to the assumptions imposed in Section~3.1 of Kalisch and B\"uhlmann\,(2007) in order to prove the consistency of the PC algorithm for Gaussian graphical models. We let the number of variables $p$ grow with the sample size $n$ and consider $p=p_n$, and also the DAG $\mathcal{G}=\mathcal{G}_n:=(V_n,E_n)$ and the distribution $P=P_n$.
	
	\begin{enumerate}
		%\item[(A3)] The distribution $P_n$ is faithful to the DAG $G_n$ for all $n$.
		%\item[(A3)] Assume that\, $\dis \sup_p \max_{1\leq a\leq p} \E\, X_a^2<\infty$.
		\item[(A3)] The dimension $p_n$ grows at a rate such that the right hand side of (\ref{conc_sup}) tends to zero as $n \to \infty$. In particular this is satisfied when $p_n = O(n^r)$ for any $0\leq r<\infty$.
		\item[(A4)] The maximum degree of the DAG $\mathcal{G}_n$, denoted by $q_n := \max_{1\leq a \leq p_n} \,| adj(\mathcal{G}_n, a)\,|$, grows at the rate of $O(n^{1-b})$, where $0<b\leq 1$.
		\item[(A5)] The distribution $P_n$ is faithful to the DAG \,$\mathcal{G}_n$\, for all $n$. In other words, for any $a, b \in V_n$ and $S \in J_{a,b}^{m_{p_n}}, $ 
		\begin{align*}
			X_a \,\,\textrm{and}\,\, X_b \,\,\,\textrm{are d-separated by} \,\, X_S  \;\; \Longleftrightarrow \;\; X_a \bigCI X_b \,\vert\, X_S \;\; \Longleftrightarrow \;\;  \rho_{0\,;\, a b|S}^* = 0 \,.
		\end{align*}
		Moreover, $\rho_{0\,;\, a b|S}^*$ values are uniformly bounded both from above and below. Formally,  %\, $C_{min}:= \dis \inf_{\substack{a, b \in V_n \\ S \in J_{a,b}^{m_{p_n}}}} \, \vert  \rho_{0\,;\, ab|S}^* \vert$ \, and\, $C_{max}:=\dis\sup_{\substack{a, b \in V_n \\ S \in J_{a,b}^{m_{p_n}}}} \, \vert  \rho_{0\,;\, ab|S}^* \vert\,.$
		\begin{align*}
			C_{min}\, :&= \, \inf_{\substack{a, b \in V_n \\ S \in J_{a,b}^{m_{p_n}} \\ \rho_{0\,;\, ab|S}^* \neq 0 }} \,   \rho_{0\,;\, ab|S}^*  \,\geq \, \lambda_{min}\,, \;\lambda_{min}^{-1}=O(n^v)  \\
			\textrm{and} \qquad C_{max}\, :&= \, \sup_{\substack{a, b \in V_n \\ S \in J_{a,b}^{m_{p_n}}}} \,   \rho_{0\,;\, ab|S}^*  \,\leq \, \lambda_{max}\,,
		\end{align*}
		where $\lambda_{min}, \lambda_{max}$ are positive constants\, and\, $0<v<1/4$.
	\end{enumerate}
	
	%Condition (A3) is a common assumption in the graphical modeling literature.
	%Condition (A3) is a reasonable moment assumption on the set of variables. 
	Condition (A3) allows the dimension to grow at any arbitrary polynomial rate of the sample size. Condition (A4) is a sparsity assumption on the underlying true DAG, allowing the maximum degree of the DAG to also grow, but at a slower rate than $n$. Since $m_p \leq q_n$, we also have $m_p = O(n^{1-b})$. Finally, Condition (A5) is the strong faithfulness assumption (Definition~1.3 in Uhler et al., 2013) on $P_n$ and is similar to condition (A4) in Kalisch and B\"uhlmann\,(2007). This essentially requires $\rho_{0\,;\, ab|S}^*$ to be bounded away from zero when the vertices $X_a$ and $X_b$ are not d-separated by $X_S$. It is worth noting that the faithfulness assumption alone is not enough to prove the consistency of the PC/PC-stable/nonPC algorithms in high-dimensional settings, and the more stringent strong faithfulness condition is required.
	
	\begin{remark}\label{rem:strf}
		For notational convenience, treat $X_a, X_b$ and $X_S$ as $X$, $Y$ and $Z$, respectively, for any $a, b \in V_n$ and $S \in J_{a,b}^{m_{p_n}}$. From Equation (\ref{cdcov equiv}) we have
		\vspace{-0.1in}
		\begin{align*}
			\CdCov^2(X,Y|Z) \;=\; \frac{1}{12}\, \E\,\big[\,d^S_{1234} \,|\, Z_1=Z, \dots, Z_4=Z \big],
		\end{align*}
		%\vspace{-0.15in}
		which implies 
		%\vspace{-0.1in}
		\begin{align*}
			\rho_{0}^* \;=\; \E\,[\CdCov^2(X,Y|Z)] \;=\; \frac{1}{12}\, \E\,\big[\,d^S_{1234} \big] \;=\; \frac{1}{12}\, \E\,\big[\,d_{1234} + d_{1243} + d_{1432} \big]\,.
		\end{align*}
		Condition (A1) implies\, $\dis \sup_p \max_{1\leq a\leq p} \E\, X_a^2<\infty$. With this and the definition of $d_{ijkl}$ in Section \ref{subsec:dcov}, it follows from some simple algebra and the Cauchy-Schwarz inequality that \,$\rho_{0}^* < \infty$. This provides a justification for the second part of Assumption (A5) that $\dis \sup_{\substack{a, b \in V_n \\ S \in J_{a,b}^{m_{p_n}}}} \,   \rho_{0\,;\, ab|S}^*  \,\leq \, \lambda_{max}$\, for some positive constant $\lambda_{max}$.
	\end{remark}

	The next theorem establishes that the nonPC algorithm consistently estimates the skeleton of a sparse high-dimensional DAG, thereby providing the necessary theoretical guarantees to our proposed methodology. It is worth noting that in the sample version of the PC-stable and hence the nonPC algorithm, all the inference is done during the skeleton estimation step. The second step that involves appropriately orienting the edges of the estimated skeleton, is purely deterministic (see Sections~4.2 and 4.3 in Colombo and Maathuis, 2014). Therefore to prove the consistency of the nonPC algorithm in estimating the equivalence class of the underlying true DAG, it is enough to prove the consistency of the estimated  skeleton. 
	
	\begin{theorem}\label{th:consistency}
		Assume that Conditions (A1)--(A5) hold. Let $\mathcal{G}_{\textrm{skel},n}$ be the true skeleton of the graph $\mathcal{G}_n$, and $\hat{\mathcal{G}}_{\textrm{skel},n}$ be the skeleton estimated by the nonPC algorithm. Then as $n \to \infty$, $\mathbb{P}\big( \hat{\mathcal{G}}_{\textrm{skel},n} = \mathcal{G}_{\textrm{skel},n} \big) \to 1$.
	\end{theorem}

	\subsection{The Nonparametric FCI Algorithm in High Dimensions}\label{subsec:methods_nonFCI}
	
	The FCI is a modification of the PC algorithm that accounts for latent and selection variables. Thus generalizations of the PC algorithm naturally extend to the FCI as well. Similar to nonPC, we propose testing for conditional independence relations nonparametrically in the \textit{sample version} of the FCI-stable algorithm (Algorithm~\ref{FCI-stable} in Appendix~\ref{sec:preliminaries}) based on $\rho_0^*$, instead of partial correlations. We coin the resulting algorithm the `nonFCI' algorithm, to emphasize that it is a generalization of parametric FCI-stable algorithms. Again, the \textit{oracle version} of the nonFCI is exactly the same as that of the FCI-stable algorithm. The difference is in the implementation of the tests for conditional independence relationships in their \textit{sample versions}. This broadens the scope of the FCI algorithm in causal structural learning for observational data in the presence of latent and selection variables when Gaussianity is not a viable assumption. More specifically, it enables capturing non-linear and non-monotone conditional dependence relationships among the variables that partial correlations would fail to detect. 
	
	Equipped with the theoretical guarantees we established for the nonPC in Section~\ref{subsec:methods_nonPC}, we establish below in Theorem \ref{th:consistency_nonFCI} the consistency of the nonFCI algorithm for general distributions in sparse high-dimensional settings. Let $\cal{H}=(V,E)$ be a DAG with the vertex set partitioned as $V = V_X \cup V_L \cup V_T$, where $V_X$ indexes the set of $p$ observed variables, $V_L$ denotes the set of latent variables and $V_T$ stands for the set of selection variables. Let $\cal{M}$ be the unique MAG over $V_X$. We let $p$ grow with $n$ and consider $p=p_n$, $\cal{H}=\cal{H}_n$ and $Q=Q_n$, where $Q$ is the distribution of $(U_1, \dots, U_p) := (X_1\,|\,V_T, \dots, X_p\,|\,V_T)$. We provide below the definition of possible-D-SEP sets (Definition~3.3 in Colombo et al., 2012).
	
	\begin{definition}\label{def:pdsep}
		Let $\cal{C}$ be a graph with any of the following edge types : $\circ \hspace{-0.04in}-\hspace{-0.04in}\circ$, $\circ\hspace{-0.07in}\rightarrow $ and $\leftrightarrow$. A possible-D-SEP\,$(X_a,X_b)$ in $\cal{C}$, denoted $\mathrm{pds}(\cal{C}, X_a, X_b)$, is defined as follows: $X_c \in$ $\mathrm{pds}(\cal{C}, X_a, X_b)$ if and only if there is a path $\pi$ between $X_a$ and $X_c$ in $\cal{C}$ such that for every subpath $\langle X_e, X_f, X_g \rangle$ of $\pi$, $X_f$ is a collider on the subpath in $\cal{C}$ or $\langle X_e, X_f, X_g \rangle$ is a triangle in $\cal{C}$.
	\end{definition}
	To prove the consistency of the nonFCI algorithm in sparse high-dimensional settings, we impose the following regularity conditions, which are similar to the assumptions imposed in Section~4 in Colombo et al.\,(2012).
	
	\begin{enumerate}
		\item[(C3)] The distribution $Q_n$ is faithful to the underlying MAG $\cal{M}_n$ for all $n$.
		
		\item[(C4)] The maximum size of the possible-D-SEP sets for finding the final skeleton in the FCI-stable algorithm (Algorithm~\ref{FCI-stable} in Appendix~\ref{sec:preliminaries}), $q_n'$, grows at the rate of $O(n^{1-b})$, where $0<b\leq 1$.
		
		\item[(C5)] For any $U_i, U_j \in \{U_1, \dots, U_{p_n}\}$ and $U_S \subseteq  \{U_1, \dots, U_{p_n}\} \backslash \{U_i, U_j\}$ with $|U_S| \leq q_n'$, assume 
		\begin{align*}
			&\inf \, \left\{\vert  \rho_0^* (U_i,U_j | U_S) \vert : \rho_0^* (U_i,U_j | U_S) \neq 0 \right\} \,\geq \, \lambda_{min}'\,, \,(\lambda_{min}')^{-1}=O(n^v)  \\
			\textrm{and} \qquad &\sup \, \vert   \rho_0^* (U_i,U_j | U_S) \vert \,\leq \, \lambda_{max}'\,,
		\end{align*}
		where $\lambda_{min}', \lambda_{max}'$ are positive constants\, and \, $0<v<1/4$.
	\end{enumerate}
	
	%Theorem \ref{th:consistency_nonFCI} validates the consistency of the nonFCI algorithm in sparse high-dimensional settings.
	
	\begin{theorem}\label{th:consistency_nonFCI}
		Suppose conditions (A1)--(A3) and (C3)--(C5) hold. Denote by $\cal{C}_n$ and $\cal{C}_n^*$ the true underlying FCI-PAG and the output of the nonFCI algorithm, respectively. Then as $n \to \infty$, $\mathbb{P}\big( \cal{C}_n^* = \cal{C}_n \big) \to 1$.
	\end{theorem}

	\section{Numerical Studies}\label{sec:num}
	
	\subsection{Performance of the nonPC Algorithm}\label{sec:num_nonPC}
	In this subsection, we compare the performances of the nonPC and the PC-stable algorithms in finding the skeleton and the CPDAG for various simulated datasets. We simulate random DAGs in the following examples and sample from probability distributions faithful to them. 
	
	\begin{example}[Linear SEM]\label{eg1}
		\textup{We first fix a sparsity parameter $s\in (0,1)$ and enumerate the vertices as $V=\{1,\dots, p\}$. We then construct a $p\times p$ adjacency matrix $\Lambda$ as follows. First initialize $\Lambda$ as a zero matrix. Next, fill every entry in the lower triangle (below the diagonal) of $\Lambda$ by independent realizations of Bernoulli random variables with success probability $s$. Finally replace each nonzero entry in $\Lambda$ by independent realizations of a Uniform$(0.1,1)$ random variable.}
		
		% \as{for later: we should include both pos and neg weights --- this is something the reviewers will likely ask for. Also, pls make a note that we should investigate the low TP rate of the methods. Finally, we should also compare with the copula method, but again this is for later (especially if the reviewers ask for it)} \ch{(okay !)}
		
		\begin{comment}
		\begin{itemize}
		\item Initialize $\Lambda$ as a zero matrix.
		\item Fill every entry in the lower triangle (below the diagonal) of $\Lambda$ by independent realizations of Bernoulli random variables with success probability $s$. 
		\item Replace each nonzero entry in $L$ by independent realizations of a Uniform$(0.1,1)$ random variable.
		\end{itemize}
		\end{comment}
		
		\textup{In this scheme, each node has the same expected degree $\bb{E}(m) = (p-1)s$, where $m$ is the degree of a node and follows a Binomial\,$(p-1,s)$ distribution. Using the adjacency matrix $\Lambda$, the data are then generated from the following linear structural equation model (SEM) : $$X = \Lambda X + \epsilon\,,$$ where $\epsilon = (\epsilon_1, \dots, \epsilon_p)$ and $\epsilon_1, \dots, \epsilon_p$ are jointly independent. To obtain samples $\{X_1^k, \dots, X_p^k\}_{k=1}^n$ on $\{X_1, \dots, X_p\}$, we first sample $\{\epsilon_1^k, \dots, \epsilon_p^k\}_{k=1}^n$ from the three following data-generating schemes. For $1\leq k \leq n$ and $1\leq i \leq p$,
			\begin{enumerate}
				\item[1.] Normal: Generate $\epsilon^k_i$\,'s independently from a standard normal distribution.
				\item[2.] Copula: Generate $\epsilon^k_i$\,'s as in (1) and then transform the marginals to a $F_{1,1}$ distribution.
				\item[3.] Mixture: Generate $\epsilon^k_i$\,'s independently from a 50-50 mixture of a standard normal and a standard Cauchy distribution.
			\end{enumerate}
		}
	\end{example}
	
	\begin{example}[Nonlinear SEM]\label{eg2}
		\textup{In this example, we first generate a $p\times p$ adjacency matrix $\Lambda$ in the similar way as in Example~\ref{eg1} and then generate the data from the following nonlinear SEM (similar to Voorman et al., 2014) : $X_i = \sum_{j\,:\,\Lambda_{ij} \neq 0} f_{ij}(X_j) + \epsilon_i$\, with\, $\epsilon_i \overset{i.i.d.}{\sim} N(0,1)$, where $1\leq j < i \leq p$. If the functions $f_{ij}$'s are chosen to be nonlinear, then the data will typically not correspond to a well-known multivariate distribution. We consider $f_{ij}(x_j) = b_{ij1} x_j + b_{ij2} x_j^2$, where $b_{ij1}$ and $b_{ij1}$ are independently sampled from $N(0,1)$ and $N(0,0.5)$ distributions, respectively.}
	\end{example}
	
	With the exception of Example~\ref{eg1}.1, the above examples are all non-Gaussian graphical models. We would thus expect the nonPC to perform better than the PC-stable in learning the unknown causal structure in these examples. For each of the four data generating methods considered above, we compare the Structural Hamming Distance (SHD) (Tsamardinos et al., 2006) between the estimated and the true skeletons of the underlying DAGs using the nonPC and PC-stable algorithms. The SHD between two undirected graphs is the number of edge additions or deletions necessary to make the two graphs match. Therefore larger SHD values between the estimated and the true skeleton correspond to worse estimates. 
	
	We consider 199 bootstrap replicates for the CdCov-based conditional independence tests in the implementation of our nonPC algorithm and the significance level $\alpha=0.05$. Table~\ref{table1} presents the average SHD for the different data generating schemes over 20 simulation runs, for different choices of $n, p$ and $\E(m)$. 
	
	\begin{table}[t]\footnotesize
		\centering
		\caption{Comparison of the average structural Hamming distances (SHD) of nonPC and PC-stable algorithms across simulation studies.}
		\label{table1}
		\begin{tabular}{ccc cc cc cc cc}
			\toprule
			&&&\multicolumn{2}{c}{Normal}&\multicolumn{2}{c}{Copula}&\multicolumn{2}{c}{Mixture}&\multicolumn{2}{c}{Nonlinear SEM}\\
			\cmidrule(r){4-5}\cmidrule(r){6-7}\cmidrule(r){8-9}\cmidrule(r){10-11}
			$n$ & $p$ & $\E(m)$\;\; & nonPC & PC-stable\;\;  & nonPC & PC-stable \;\; & nonPC & PC-stable \;\; & nonPC & PC-stable\\
			\hline
			50 & 9 & 1.4 & 3.35 & 3.05 & 5.55 & 5.75 & 3.8 &    3.5  & 2.9 & 3.7 \\
			100 & 27  & 2.0 & 14.55 & 11.00 & 25.6 & 28.6 & 17.75 & 18.00  & 15.05 & 20.05\\
			150 & 81 & 2.4 & 53.70 & 43.45 & 97.3 & 121.3 & 69.05 & 77.75  & 62.583 & 95.083 \\
			200 & 243 & 2.8 & 186.2 & 183.4 & 331.00 & 471.45 & 250.3 & 336.1 & 213.70 & 375.45 \\
			%250 & 729 & 3.2 &  \\
			\toprule
		\end{tabular}
		\\
	\end{table}
	
	The results in Table \ref{table1} demonstrate that the nonPC performs nearly as good as the PC-stable for the Gaussian data example, in terms of the average SHD. But for each of the non-Gaussian data examples the nonPC performs better than the PC-stable in estimating the true skeleton of the underlying DAGs. The improvement in SHD becomes more substantial as the dimension grows. The superior performance of the nonPC over PC-stable for the non-Gaussian graphical models is expected, as the characterization of conditional independence by partial correlations is only valid under the assumption of joint Gaussianity.

	\subsection{Performance of the nonFCI algorithm}\label{sec:num_nonFCI}
	In this subsection, we compare the performances of the nonFCI and the FCI-stable algorithms over various simulated datasets.
	We first generate random DAGs as in Examples~\ref{eg1} and \ref{eg2}. To assess the impact of latent variables, we randomly define half of the variables with no parents and at least one child as latent. We do not consider selection variables. We run both the nonFCI and the FCI-stable algorithms on the above data examples with $n=200$, $p=\{10,20,30,100,200\}$ and $\alpha=0.01$, using 199 bootstrap replicates for the CdCov-based conditional independence tests. We consider 20 simulation runs for each of the data generating models. Table~\ref{table2} reports the average SHD between the estimated and true PAG skeleton by the nonFCI and FCI-stable algorithms.
	
	\begin{table}[t]\footnotesize
		\centering
		\caption{Comparison of the average structural Hamming distances (SHD) of nonFCI and FCI-stable algorithms across simulation studies.}
		\label{table2}
		\begin{tabular}{cc cc cc cc cc }
			\toprule
			&&\multicolumn{2}{c}{Normal}&\multicolumn{2}{c}{Copula}&\multicolumn{2}{c}{Mixture}&\multicolumn{2}{c}{Nonlinear SEM}\\
			\cmidrule(r){3-4}\cmidrule(r){5-6}\cmidrule(r){7-8} \cmidrule(r){9-10}
			$p$ & $\E\,N$ & nonFCI & FCI-stable & nonFCI & FCI-stable & nonFCI & FCI-stable & nonFCI & FCI-stable \\
			\hline
			10 & 2.0 & 7.15 & 7.60 &  1.3 & 1.8 & 5.65 &  6.80 & 7.15 & 8.20 \\
			20  & 2.0 & 14.55 & 17.60 & 4.55 & 6.85 & 13.65 & 18.55 & 19.0 & 20.8 \\
			30 & 2.0 & 27.65 & 33.95 & 5.25 & 10.15 & 19.3 & 27.8 & 33.40 & 37.85 \\
			100 & 3.0 & 109.30 & 150.35 & 26.95 & 60.05 & 62.25 & 111.10 & 115.2 & 149.0 \\
			200 & 3.0 & 287.75 & 371.40 & 76.733 & 157.267 & 136.05 & 255.10 & 289.6 &  354.1 \\
			\toprule
		\end{tabular}
		\\
	\end{table}
	
	\begin{comment}
	\begin{table}[H]\small
	\centering
	\caption{Performance comparison of the nonFCI and FCI-stable algorithms}
	\label{table3}
	\begin{tabular}{cc cccc cccc }
	\toprule
	&&\multicolumn{4}{c}{Mixture}&\multicolumn{4}{c}{Nonlinear SEM}\\
	\cmidrule(r){3-6}\cmidrule(r){7-10}
	& & \multicolumn{2}{c}{nonFCI}  &  \multicolumn{2}{c}{FCI-stable} & \multicolumn{2}{c}{nonFCI} &  \multicolumn{2}{c}{FCI-stable} \\
	\cmidrule(r){3-4}\cmidrule(r){5-6}\cmidrule(r){7-8} \cmidrule(r){9-10}
	$p$ & $\E\,N$ & TP & FP & TP & FP & TP & FP & TP & FP \\
	\hline
	10 & 2.0 & 4.4   &  1.4  &   7.2   &  3.1 & 5.7  &   3.4  &   6.4     & 6.2 \\
	20  & 2.0 &  2.3  &   6.8  &   5.0  &  15.0 & 4.2  &  14.2  &   4.3    & 19.1 \\
	30 & 2.0 & 1.7  &  12.0  &   3.0  &  28.1 & 2.2  &  22.5  &   2.9  &  30.4 \\
	100 & 3.0 & 0.5 & 42.2 & 1.4 & 140.2 & 1.8 & 75.8 & 2.8 & 126.3 \\
	200 & 3.0 & 1.0 & 80.1 & 2.7 & 313.7 & 2.6 & 140.7 & 4 & 276  \\
	\toprule
	\end{tabular}
	\\
	\end{table}
	\end{comment}
	
	The results in Table~\ref{table2} demonstrate that in both the Gaussian and non-Gaussian examples the nonFCI algorithm outperforms the FCI-stable in estimating the true PAG skeleton.

	\subsection{Real data example}\label{sec:num_real}
	
	To demonstrate the flexibility of the proposed framework, we apply the nonPC algorithm to the Montana Economic Outlook Poll dataset. The poll was conducted in May 1992 where a random sample of 209 Montana residents were asked whether their personal financial status was worse, the same or better than a year ago, and whether they thought the state economic outlook was better than the year before. Accompanying demographic information on the respondents' age, income, political orientation, and area of residence in the state were also recorded. We obtained the dataset from the Data and Story Library (DASL), available at \url{https://math.tntech.edu/e-stat/DASL/page4.html}. The study comprises of the following 7 categorical variables: AGE = 1 for under 35, 2 \,for 35-54, 3 \,for 55 and over;\, SEX = 0 \, for male, 1\, for female;\, INC = yearly income: 1 for under \$20K, 2\, for \$20-35K, 3\, for over \$35K;\, POL = 1 for Democrat, 2 \,for Independent, 3 \,for Republican;\, AREA = 1 for Western, 2 \,for Northeastern, 3 \,for Southeastern Montana;\, FIN (= Financial status): 1 for worse, 2\, for same, 3\, for better than a year ago; \,and\, STAT (= State economic outlook): 1 for better, 0\, for not better than a year ago.
	
	After removing the cases with missing values, we are left with $n=163$ samples. Since all the variables are categorical, the Gaussianity assumption is outrightly violated. Thus we would expect the nonPC to perform better than the PC-stable in learning the true causal structure among the variables in this case. Figure~\ref{fig_real} below presents the CPDAGs estimated by the nonPC and PC-stable algorithms at a significance level $\alpha=0.1$ with 199 bootstrap replicates for the CdCov-based conditional independence tests.
	
	It is quite intuitive that age and sex are likely to affect the income; one's financial status and the area of residence might also influence their political inclination; and improvements or downturns in the state economic outlook might impact an individual's financial status. The CPDAG estimated by the nonPC algorithm in Figure~\ref{fig_real}a \,affirms such common-sense understandings of these causal influences. However, in the CPDAG estimated by the PC-stable in Figure~\ref{fig_real}b, the edge between age and income is missing. Also the directed edges POL $\rightarrow$ AREA and POL $\rightarrow$ FIN seem to make little sense in this case.
	
	\begin{comment}
	\begin{enumerate}
	\item AGE = 1 under 35, 2 for 35-54, 3 for 55 and over
	\item SEX = 0 male, 1 female
	\item INC = yearly income : 1 under \$20K, 2 for \$20-35K, 3 over \$35K
	\item POL = 1 Democrat, 2 Independent, 3 Republican
	\item AREA = 1 Western, 2 Northeastern, 3 Southeastern Montana
	\item FIN = Financial status : 1 worse, 2 same, 3 better than a year ago
	\item STAT = State economic outlook : 1 better, 2 not better than a year ago.
	\end{enumerate}
	\end{comment}
	
	\begin{figure}[t]
		\vskip -0.1in
		\centering
		\subfloat[nonPC]{\includegraphics[width=0.45\linewidth]{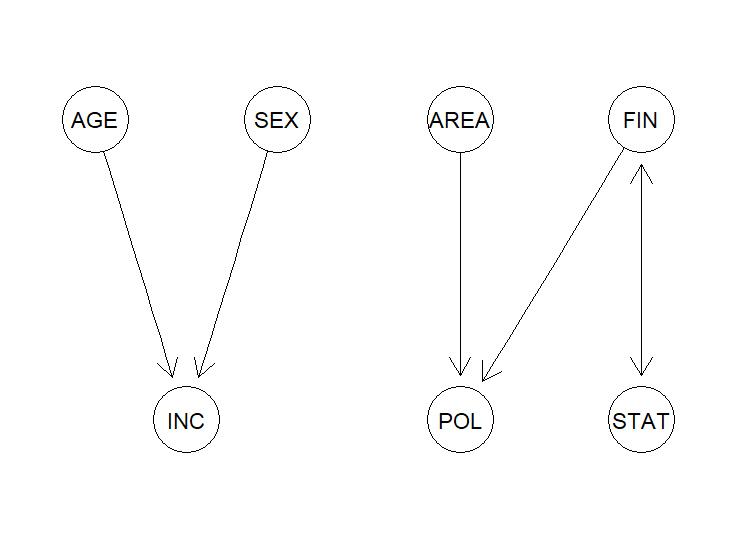}}
		\qquad \qquad
		\subfloat[PC-stable]{\includegraphics[width=0.45\linewidth]{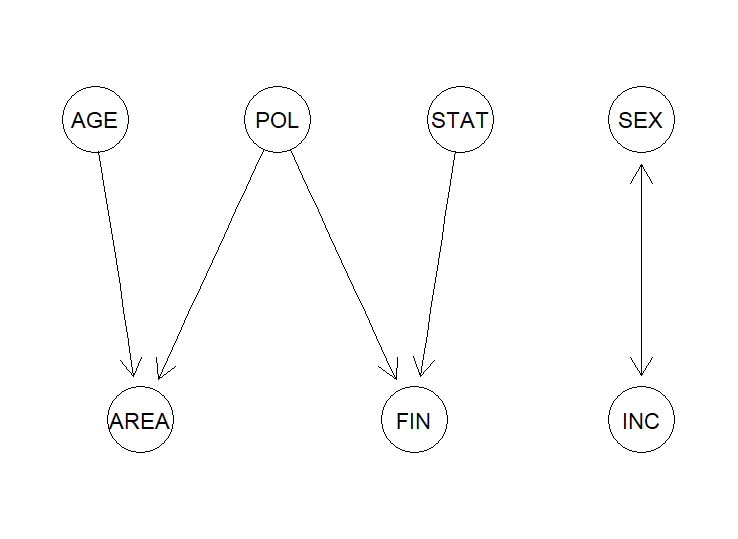}}
		\caption{CPDAGs estimated by the nonPC and PC-stable algorithms for the Montana poll dataset.}\label{fig_real}
		\vskip -0.1in
	\end{figure}

	\section{Proofs of the Theoretical Results}\label{sec:appendix}
	
	In this section, we provide detailed technical proofs of the theoretical results presented in the paper. We first state and prove a concentration inequality that will be used in the proof of Theorem~\ref{th:conc}. 
	
	\begin{lemma}\label{sup_lemma1}
		Consider a U-statistic $U_n = U(X_1, \dots, X_n) = \binom{n}{m}^{-1} \dis \sum_{i_1< \dots <i_m} h(X_{i_1}, \dots, X_{i_m})$ with a symmetric kernel $h$ such that $\E U_n = \E h(X_1, \dots, X_m) = \theta$. Further suppose $\vert h(X_1, \dots, X_m) \vert \leq M$ for some $M>0$. Then for any $\epsilon>0$, we have
		\begin{align*}
			\mathbb{P}\left(\vert U_n - \theta \vert > \epsilon \right)\;\leq\; 2\,\exp\left( -\frac{\epsilon^2 \,k}{2M^2} \right)\,, 
		\end{align*}
		where $k:=\lfloor \frac{n}{m} \rfloor$.
	\end{lemma}

	\begin{proof}[Proof of Lemma \ref{sup_lemma1}]
		Define $$W(X_1, \dots, X_n) \;:=\; \frac{1}{k}\,\big[ h(X_1,\dots, X_m)\,+\, h(X_{m+1},\dots, X_{2m}) \,+\, \cdots \,+\, h(X_{km-m+1}, \dots, X_{km}) \, \big]\,.$$ Then, following Section 5.1.6 in Serfling (1980), we can write 
		\begin{equation}\label{supp1}
			U_n \;=\; \frac{1}{n!}\dis \sum_{\pi} W(X_{i_1}, \dots, X_{i_n})\,,
		\end{equation}
		where $\dis\sum_{\pi}$ denotes summation over all \,$n!$\, permutations $(i_1, \dots, i_n)$ of $(1,2,\dots, n)$. Thus $U_n$ can be expressed as an average of $n!$ terms, each of which is an average of $k$ i.i.d. random variables. Using Markov's inequality, convexity of the exponential function and Jensen's inequality, we have, for any $t>0$,
		\begin{align}\label{supp2}
			\begin{split}
				\mathbb{P}\left(U_n - \theta > \epsilon \right)\;&=\; \mathbb{P}\big(\exp \big( t\,(U_n - \theta)\big) > \exp \left(t\epsilon\right) \big)\\  
				&\leq\;\exp(-t\epsilon) \exp(-t\theta) \,\E\,\left[\exp \big(t\,U_n\big)\right]\\
				&=\; \exp(-t\epsilon) \exp(-t\theta) \,\E \left[\exp \left(t\,\frac{1}{n!}\dis \sum_p W(X_{i_1}, \dots, X_{i_n})\right)\right]\\
				&\leq \; \exp(-t\epsilon) \exp(-t\theta) \, \frac{1}{n!}\dis \sum_{\pi} \E \left[\exp \big( t\, W(X_{i_1}, \dots, X_{i_n})\big) \right]\\
				&= \; \exp(-t\epsilon) \exp(-t\theta) \, \left[\E\left( \exp \left(\frac{t}{k}\, h \right)\right)\right]^k\\
				&= \; \exp(-t\epsilon)\, \E^k \left[ \exp \Big(\frac{t}{k} \left( h - \theta\right)\Big)\right]\,,
			\end{split}
		\end{align}
		where, for notational simplicity, we use $h$ to denote $h(X_1, \dots, X_m)$. Using Hoeffding's Lemma, we have from (\ref{supp2})
		\begin{align*}
			\mathbb{P}\left(U_n - \theta > \epsilon \right)\;\leq\; \exp\left( -t\epsilon \,+\, k\, \frac{1}{8}\,\frac{t^2}{k^2}\,(2M)^2 \right)\;=\;\exp\left( -t\epsilon \,+\, \frac{t^2 M^2}{2k} \right)\,. 
		\end{align*}
		Symmetrically, we obtain
		\begin{equation}\label{supp3}
			\mathbb{P}\left(\vert U_n - \theta \vert > \epsilon \right)\;\leq\; 2\,\exp\left( -t\epsilon \,+\, \frac{t^2 M^2}{2k} \right)\,. 
		\end{equation}
		The right hand side of (\ref{supp3}) is minimized at $t=\epsilon \,k/M^2$. Therefore, choosing $t=\epsilon \,k/M^2$, we get 
		\begin{align*}
			\mathbb{P}\left(\vert U_n - \theta \vert > \epsilon \right)\;\leq\; 2\,\exp\left( -\frac{\epsilon^2 \,k}{2M^2} \right)\,. 
		\end{align*}
		
	\end{proof}

	\begin{proof}[Proof of Theorem \ref{th:conc}]
		
		When $|S|=0$, it can be shown in similar lines of Theorem 1 in Li et al.\,(2012) that for any\, $\epsilon > 0$, there exist positive constants $A$, $B$ and\, $\gamma \in (0,1/4)$\, such that 
		\vspace{-0.1in}
		\begin{align*}
			\mathbb{P}\left(\vert\, \widehat{\rho}^{\;*}(X_a, X_b | X_S) - \rho_0^*\,(X_a, X_b | X_S) \,\vert > \epsilon \right) \;\leq \;  O\Big( 2\,\exp\left( - A\,n^{1-2\gamma}\,\epsilon^2  \right) \,+\, n\, \exp\big( - B \, n^{\gamma}\big)\Big)\,.
		\end{align*}
		Now consider the case $0<|S|\leq m_p$. For notational convenience, we treat $X_a, X_b$ and $X_S$ as $X$, $Y$ and $Z$, respectively. 
		
		Denote $\delta_Z := \CdCov^2(X, Y \vert Z)$. Then, $\rho_{0}^* = \E[\delta_Z]$. Recall that 
		\vspace{-0.15in}
		\begin{align}\label{newnum1}
			\begin{split}
				&\widehat{\rho}^{\;*}(X, Y | Z) \;:=\; \frac{1}{n} \dis \sum_{u=1}^n \CdCov^2_n(X,Y|Z_u)\, \;:=\; \frac{1}{n} \dis \sum_{u=1}^n \Delta_{i,j,k,l;u}\,,\\
				\textrm{where}\qquad &\Delta_{i,j,k,l;u} \;:=\; \dis \sum_{i,j,k,l} \frac{K_{iu}\, K_{ju}\, K_{ku}\, K_{lu}}{12\,\big(\sum_{i=1}^n K_{iu}\big)^4}\; d^S_{ijkl}\,.
			\end{split}
		\end{align}
		From (\ref{newnum1}), we have
		\begin{align}\label{supp13.5}
			\begin{split}
				\E\,\big[\CdCov^2_n(X,Y|Z_u) \vert Z \big] \;&=\; \frac{1}{12} \,\E\,\big[\,d^S_{1234} \,|\, Z_1=Z_u, \dots, Z_4=Z_u \big] \, \dis \sum_{i,j,k,l} K_{iu}\,K_{ju}\,K_{ku}\,K_{lu}\,/\,\left( \sum_{i=1}^n K_{iu}\right)^4 \\
				&=\; \frac{1}{12} \,\E\,\big[\,d^S_{1234} \,|\, Z_1=Z_u, \dots, Z_4=Z_u \big] \;=\; \delta_{Z_u}\,,
			\end{split}
		\end{align}
		%\vspace{-0.4in}
		where the last equality  follows from Lemma 1 in Wang et al.\,(2015). Together, (\ref{newnum1}) and (\ref{supp13.5}) imply $\E\,[\,\widehat{\rho}^{\;*}] \;=\; \rho_{0}^*$. 
		
		Now consider the truncation
		\begin{align}\label{newnum3}
			\begin{split}
				%\frac{1}{12}\, d^S_{i,j,k,l} \;&=\; \frac{1}{12}\, d^{S1}_{i,j,k,l} \;+\; \frac{1}{12}\, d^{S2}_{i,j,k,l}\\
				%&:=\; \frac{1}{12}\, d^S_{i,j,k,l}\, {\bf 1}\left(\left\vert \frac{1}{12}\, d^S_{i,j,k,l} \right\vert \leq M \right) \;+\; \frac{1}{12}\, d^S_{i,j,k,l} \,{\bf 1}\left(\left\vert \frac{1}{12}\, d^S_{i,j,k,l} \right\vert > M \right) \\
				\rho_0^* \;&=\; \rho_{01}^* \,+\, \rho_{02}^* \\
				&:= \; \E\,\left[  \frac{1}{12}\, d^S_{i,j,k,l}\, {\bf 1}\left(\left\vert \frac{1}{12}\, d^S_{i,j,k,l} \right\vert \leq M \right)  \right] \;+\; \E\,\left[  \frac{1}{12}\, d^S_{i,j,k,l} \,{\bf 1}\left(\left\vert \frac{1}{12}\, d^S_{i,j,k,l} \right\vert > M\right)  \right] \,,
			\end{split}
		\end{align}
		where $M>0$ will be specified later. Then, using triangle inequality,
		\begin{align}\label{newnum4}
			\begin{split}
				\mathbb{P}\left(\vert \widehat{\rho}^{\;*} - \rho_{0}^* \vert > \epsilon \right) \;&=\; \mathbb{P}\left(\left\vert \frac{1}{n} \dis \sum_{u=1}^n \big( \sum_{i,j,k,l} \Delta_{i,j,k,l; u} - \rho_{0}^* \big) \right\vert  > \epsilon \right)\\
				& \leq \; \mathbb{P}\left(\left\vert \frac{1}{n} \dis \sum_{u=1}^n \left( \sum_{i,j,k,l} \Delta_{i,j,k,l; u} \, {\bf 1}\left(\left\vert \frac{1}{12}\, d^S_{i,j,k,l} \right\vert \leq M\right) - \rho_{01}^* \right) \right\vert  > \epsilon/2 \right) \\ 
				& \;\; +\, \mathbb{P}\left(\left\vert \frac{1}{n} \dis \sum_{u=1}^n \left( \sum_{i,j,k,l} \Delta_{i,j,k,l; u} \, {\bf 1}\left(\left\vert \frac{1}{12}\, d^S_{i,j,k,l} \right\vert > M\right) - \rho_{02}^* \right) \right\vert  > \epsilon/2 \right)\\
				& =: \mathrm{I} \,+\, \mathrm{II}\,.
			\end{split}
		\end{align}
		Clearly from (\ref{newnum1}) we have $\vert \Delta_{i,j,k,l; u} \vert \leq M$ when $\left\vert \frac{1}{12}\, d^S_{i,j,k,l} \right\vert \leq M$. With this observation, we have
		\begin{align}\label{newnum5}
			\mathrm{I} \;& \leq \; 2\,\exp\left( -\frac{n\,\epsilon^2 }{8\,M^2} \right)\,,
		\end{align}
		which follows from Lemma \ref{sup_lemma1} by setting $m=1$, $k=\lfloor n \rfloor$ and $\epsilon=\epsilon/2$. Choosing $M=c\,n^{\gamma}$ for $\gamma \in (0,1/4)$ and some positive constant $c$, it follows from (\ref{newnum5}) that
		\begin{align}\label{newnum6}
			\mathrm{I} \;& \leq \; 2\,\exp\left( - C_1\,n^{1-2\gamma}\,\epsilon^2  \right)\,,
		\end{align}
		for some $C_1>0$.
		
		Now to find a suitable upper bound for II, note that a simple application of triangle inequality yields 
		\begin{align}\label{supp8}
			\begin{split}
				\frac{\epsilon}{2} \; &< \; \left\vert\, \frac{1}{n} \dis \sum_{u=1}^n \sum_{i,j,k,l} \Delta_{i,j,k,l; u} \, {\bf 1}\left(\left\vert \frac{1}{12}\, d^S_{i,j,k,l} \right\vert > M\right) - \rho_{02}^* \,\right\vert \\
				&\leq \; \left\vert\, \frac{1}{n} \dis \sum_{u=1}^n \sum_{i,j,k,l} \Delta_{i,j,k,l; u} \, {\bf 1}\left(\left\vert \frac{1}{12}\, d^S_{i,j,k,l} \right\vert > M\right) \,\right\vert \,+\, \vert \rho_{02}^* \vert\,.
			\end{split}
		\end{align}
		For the choice of $M=c\,n^{\gamma}$, we have 
		\begin{align}\label{supp9}
			\rho_{02}^* \;=\; \E\,\left[  \frac{1}{12}\, d^S_{i,j,k,l} \,{\bf 1}\left(\left\vert \frac{1}{12}\, d^S_{i,j,k,l} \right\vert > M\right)  \right] \;<\; \frac{\epsilon}{4}
		\end{align}
		for sufficiently large $n$ (see, for example, Exercise 6 in Chapter 5, Resnick\,(1999)). Combining (\ref{supp8}) and (\ref{supp9}), we get
		\begin{align*}
			&\left\{ \left\vert\, \frac{1}{n} \dis \sum_{u=1}^n \sum_{i,j,k,l} \Delta_{i,j,k,l; u} \, {\bf 1}\left(\left\vert \frac{1}{12}\, d^S_{i,j,k,l} \right\vert > M\right) - \rho_{02}^* \,\right\vert > \epsilon/2 \right\} \\
			&\subseteq \; \left\{ \left\vert\, \frac{1}{n} \dis \sum_{u=1}^n \sum_{i,j,k,l} \Delta_{i,j,k,l; u} \, {\bf 1}\left(\left\vert \frac{1}{12}\, d^S_{i,j,k,l} \right\vert > M\right) \,\right\vert > \epsilon/4 \right\}\\
			& \subseteq \; \left\{ \left[\, \left\vert \frac{1}{12}\, d^S_{i,j,k,l} \right\vert > M \right] \;\textrm{for some}\,\, 1\leq i, j, k, l \leq n \right\},
		\end{align*}
		which implies 
		\begin{align}\label{supp10}
			\begin{split}
				& \mathbb{P}\left( \left\vert\, \frac{1}{n} \dis \sum_{u=1}^n \sum_{i,j,k,l} \Delta_{i,j,k,l; u} \, {\bf 1}\left(\left\vert \frac{1}{12}\, d^S_{i,j,k,l} \right\vert > M\right) - \rho_{02}^* \,\right\vert > \epsilon/2 \right) \\
				&\leq \; \mathbb{P}\left( \left\vert\, \frac{1}{n} \dis \sum_{u=1}^n \sum_{i,j,k,l} \Delta_{i,j,k,l; u} \, {\bf 1}\left(\left\vert \frac{1}{12}\, d^S_{i,j,k,l} \right\vert > M\right) \,\right\vert > \epsilon/4 \right)\\
				& \leq \; n^4\; \mathbb{P}\left(\left\vert \frac{1}{12}\, d^S_{i,j,k,l} \right\vert > M \right)\,.
			\end{split}
		\end{align}
		This is because if $\left\vert \frac{1}{12}\, d^S_{i,j,k,l} \right\vert \leq M$ for all $1\leq i, j, k, l \leq n$, then
		\[
		n^{-1} \sum_{u=1}^n \sum_{i,j,k,l} \Delta_{i,j,k,l; u} \, {\bf 1}\left(\left\vert \frac{1}{12}\, d^S_{i,j,k,l} \right\vert > M\right) \\= 0.
		\]
		Under Condition (A1), Lemma 2 in the supplementary materials of Wen et al.\,(2018) proves that there exists $s>0$ for which $\E \left[\exp\big( s \,\big\vert \,d^S_{1234}  \big\vert\big)\right]$ is finite. Using Markov's inequality, we have
		\begin{align}\label{supp12}
			\begin{split}
				\mathbb{P}\left(\left\vert \frac{1}{12}\, d^S_{i,j,k,l} \right\vert > M \right) \;&\leq \; \mathbb{P}\left(\,\exp \left(s\,\left\vert \frac{1}{12}\, d^S_{i,j,k,l} \right\vert \right) >\, \exp(sM) \right)\\
				&\leq \; \exp(-sM)\, \E\,\left[ \exp \left(s\,\left\vert \frac{1}{12}\, d^S_{i,j,k,l} \right\vert \right) \right]\\
				& \leq \; C_2\,\exp(-sM) \;\leq \; C_2\,\exp(-s_1 \,n^{\gamma})\,,
			\end{split}
		\end{align}
		for some positive constants $C_2$ and $s_1$, where last line uses the fact that $M=c\,n^{\gamma}$. Combining (\ref{supp10}) and (\ref{supp12}), we have
		\begin{align}\label{newnum7}
			\mathrm{II} \;\leq \; C_2\,n^4\,\exp(-s_1 \,n^{\gamma})\,.
		\end{align}
		Finally, combining (\ref{newnum4}), (\ref{newnum6}) and (\ref{newnum7}), we get 
		\begin{align*}
			\mathbb{P}\left(\vert \widehat{\rho}^{\;*} - \rho_0^* \vert > \epsilon/2 \right) \;\leq \; 2\,\exp\left(-C_1 n^{1-2\gamma} \epsilon^2\right) \,+\, C_2\,\, n^4 \,\exp\left(- s_1 n^{\gamma}  \right)\,,
		\end{align*}
		for some positive constants \,$\gamma, C_1, C_2$ and $s_1$. This completes the proof of the theorem.

	\end{proof}

	\begin{proof}[Proof of Theorem \ref{th:conc_sup}]
		The first inequality in Theorem \ref{th:conc_sup} simply follows by observing the fact that for any generic random sequence $\{X_n\}_{n=1}^{\infty}$ and any $\epsilon > 0$, 
		\begin{align*}
			P(\vert X_n \vert > \epsilon ) \; \leq \; P\big(\sup_n \,\vert X_n \vert > \epsilon \big)
		\end{align*}
		for all $n\geq 1$, which in turn implies 
		\begin{align*}
			\sup_n \,P(\vert X_n \vert > \epsilon ) \; \leq \; P\big(\sup_n \vert X_n \vert > \epsilon \big)\,.
		\end{align*}
		The second inequality follows from union bound and Theorem \ref{th:conc}.
	\end{proof}
	
	\begin{proof}[Proof of Theorem \ref{th:consistency}]
		Denote by $E_{a b|S}$ the event that ``an error occurs while testing for $X_a \bigCI X_b \,|\, X_S$" for $a, b \in  V$ and $S \in J_{a,b}^{m_{p_n}}$. Then 
		\begin{align}\label{supp100}
			%\begin{split}
			\mathbb{P}(\,\textrm{an error occurs in the nonPC algorithm}\,) \; \leq \; \mathbb{P}\Bigg( \bigcup_{\substack{a, b \,\in \, V \\ S \in J_{a,b}^{m_{p_n}}}} E_{a b|S} \Bigg) \;\lesssim \; p_n^{m_{p_n} + 2}\, \mathbb{P}(E_{a b|S})\,,
			%\end{split}
		\end{align}
		which is essentially due to the union bound. Now, we can write $E_{ab|S} = E_{ab|S}^{\,\mathrm{I}} \cup E_{ab|S}^{\,\mathrm{II}}$, where
		\begin{align*}
			&\textrm{(Type I error)} \qquad E_{ab|S}^{\,\mathrm{I}} \,:\, \vert \hat{\rho}^*_{ab|S} \vert > \xi_n \qquad \textrm{when} \;\; \rho_{0\,;\, ab|S}^* = 0 \\
			\textrm{and} \qquad &\textrm{(Type II error)} \qquad E_{ab|S}^{\,\mathrm{II}} \,:\, \vert \hat{\rho}^*_{ab|S} \vert \leq \xi_n \qquad \textrm{when} \;\; \rho_{0\,;\, ab|S}^* > 0\,.
		\end{align*}
		Then by the using triangle inequality
		\begin{align}\label{supp101}
			\begin{split}
				\mathbb{P}(E_{ab|S}^{\,\mathrm{I}}) \;&= \; \mathbb{P}(\vert\, \hat{\rho}^*_{ab|S} \,\vert > \xi_n) \; = \; \mathbb{P}\big(\vert\, \hat{\rho}^*_{ab|S} - \rho_{0\,;\, ab|S}^* + \rho_{0\,;\, ab|S}^* \,\vert > \xi_n \big)\\
				& \leq \; \mathbb{P}\big(\vert\, \hat{\rho}^*_{ab|S} - \rho_{0\,;\, ab|S}^* \,\vert > \xi_n - C_{max} \big)\\
				& \lesssim \; 2\,\exp\big(-A\, n^{1-2\gamma} (\xi_n - C_{max})^2\big) \,+\,  n^4\, \exp\big(- B n^{\gamma}  \big)\,
			\end{split}
		\end{align}
		for positive constants $A, B$ and $\gamma \in (0,1/4)$,\, where the last inequality follows from Theorem \ref{th:conc_sup}. Similarly, using the definition of $C_{min}$ and the identity $|a| - |b| \leq \vert a-b \vert$ for $a,b \in \mathbb{R}$, we have
		\begin{align}\label{supp102}
			\begin{split}
				\mathbb{P}\left(E_{ab|S}^{\,\mathrm{II}}\right) \;&= \; \mathbb{P}(\vert\, \hat{\rho}^*_{ab|S} \,\vert \leq \xi_n) \; =  \; \mathbb{P}(- \vert\, \hat{\rho}^*_{ab|S} \,\vert \geq - \xi_n) \\
				&= \; \mathbb{P}(\vert \rho_{0\,;\, ab|S}^* \vert - \vert\, \hat{\rho}^*_{ab|S} \vert \geq \vert \rho_{0\,;\, ab|S}^* \vert - \xi_n)\\
				& \leq \; \mathbb{P}(\vert \rho_{0\,;\, ab|S}^* - \hat{\rho}^*_{ab|S} \vert \geq C_{min} - \xi_n)\\
				& \lesssim \; 2\,\exp\big(-A\, n^{1-2\gamma} (\xi_n - C_{min})^2\big) \,+\,  n^4\, \exp\big(- B n^{\gamma}  \big)\,.
			\end{split}
		\end{align}
		Again the last inequality follows from Theorem \ref{th:conc_sup}. Combining equations (\ref{supp100}), (\ref{supp101}) and (\ref{supp102}), we have
		\begin{align*}
			& \; \mathbb{P}\,(\,\textrm{an error occurs in the nonPC algorithm}\,)\\
			&= \; O\Big(\,p_n^{m_{p_n} +2} \big[\, 2\,\exp\big(-A\, n^{1-2\gamma} (\xi_n - C_{max})^2\big) \,+\, 2\,\exp\big(-A\, n^{1-2\gamma} (\xi_n - C_{min})^2\big] \\
			& \qquad \qquad +\,  n^4 \exp\big(- B\, n^{\gamma}  \big)  \big] \Big)\\
			&= \; o(1)\,,
		\end{align*}
		where the last step follows from the fact that $\gamma \in (0,1/4)$ and Assumption (A5). This implies that as $n \to \infty$,
		\begin{align*}
			\mathbb{P}\left( \hat{G}_{\textrm{skel},n} = G_{\textrm{skel},n} \right) \;&= \; 1 \,-\, \mathbb{P}\,(\,\textrm{an error occurs in the nonPC algorithm}\,) \\
			& \to \; 1\,.
		\end{align*}
		
	\end{proof}

	\begin{proof}[Proof of Theorem \ref{th:consistency_nonFCI}]
		The proof follows similar lines of the proof of Theorem 4.2 in Colombo et al.\,(2012), replacing Lemma 1.4 in their supplement by Theorem \ref{th:conc_sup} in our paper.
		
	\end{proof}

	\section{Discussion}\label{sec:discussion}
	We proposed nonparametric variants of the widely popular PC-stable and FCI-stable algorithms, which employ conditional distance covariance (CdCov) to test for conditional independence relationships in their sample versions. Our proposed algorithms broaden the applicability of the PC/PC-stable and FCI/FCI-stable algorithms to general distributions over DAGs, and enable taking into account non-linear and non-monotone conditional dependence among the random variables, which partial correlations fail to capture. We show that the high-dimensional consistency of the PC-stable and FCI-stable algorithms carry over to general distributions over DAGs when we implement CdCov-based nonparametric tests for conditional independence. Our consistency results only require mild moment and tail conditions on the set of variables, without requiring any strict distributional assumptions. 
	
	% \as{I added somethign about tuning par below, please check and complete}
	There are several intriguing potential directions for future research. First, it is generally difficult to select the tuning parameter (i.e., the significance threshold for the CdCov test) in causal structure learning. One possible strategy is to use ideas based on \textit{stability selection} (Meinshausen and B\"uhlmann, 2010; Shah and Samworth, 2013). By assessing the stability of the estimated graphs in multiple subsamples, this strategy allows us to choose the tuning parameter in order to control the false positive error. However, the repeated subsampling increases the computational burden. 
	Second, the computational and sample complexities of the PC and FCI algorithms (and hence those of the nonPC and nonFCI) scale with the maximum degree of the DAG, which is assumed to be small relative to the sample size. However, in many applications one encounters sparse graphs containing a small number of highly connected `hub' nodes. In such cases, Sondhi and Shojaie\,(2019) proposed a low-complexity variant of the PC algorithm, namely the {\it reduced PC} (rPC) algorithm, that exploits the local separation property of large random networks (Anandkumar et al., 2012). The rPC is shown to consistently estimate the skeleton of a high-dimensional DAG by conditioning only on sets of small cardinality. In this light, it would be intriguing to develop computationally faster variants of the nonPC and nonFCI in future by exploiting the idea of local separation.

	\vspace{0.2in}

	\appendix
	
	\section{Preliminaries}\label{sec:preliminaries}
	
	%The supplementary material is organized as follows. In Section \ref{sec:preliminaries}, we present the pseudocodes of the oracle versions of the PC-stable and FCI-stable algorithms. In Section \ref{sec:appendix} we provide detailed technical proofs of the theoretical results presented in the main paper.
	
	%\section{Preliminaries}\label{sec:preliminaries}
	
	For the sake of completeness, we illustrate in this section the pseudocodes of the oracle versions of the PC-stable and FCI-stable algorithms.
	
	Algorithms \ref{PC-stable1} presents the pseudocode of the oracle version of Step 1 of the PC-stable algorithm (Algorithm 4.1 of Colombo and Maathuis, 2014), which estimates the skeleton of the underlying DAG. Algorithm \ref{PC-stable2} presents the pseudocode of Step 2 of the PC-stable algorithm (Algorithm 2 of Kalisch and B\"uhlmann, 2007), that extends the skeleton estimated in Step 1 to the CPDAG.  
	Algorithm \ref{FCI-stable} presents the pseudocode of the FCI-stable algorithm (Section 4.4 in Colombo and Maathuis, 2014). It implements Algorithm \ref{FCI-stable1} to obtain an initial skeleton of the underlying PAG, Algorithm \ref{FCI-stable2} to orient the v-structures, and finally Algorithm \ref{FCI-stable3} to obtain the final skeleton that the FCI-stable returns.
	
	%\as{somehow the spacing of alg 1-2 are different -- need to make them consistent}
	
	\begin{algorithm}[!ht]
		\caption{Step 1 of the PC-stable algorithm (oracle version)} \label{PC-stable1}
		
		\begin{algorithmic}%[1]
			\State \textbf{Require} : Conditional independence information among all variables in $V$, and an ordering order$(V)$ on the variables.
			\State Form the complete undirected graph $\cal{C}$ on the vertex set $V$.
			\State Let $l=-1$;
			\Repeat
			\State $l=l+1$;
			\For{all vertices $X_a$ in $\cal{C}$}
			\State let $u(X_a) = adj(\cal{C}, X_a)$
			\EndFor
			\Repeat
			\State Select a (new) ordered pair of vertices $(X_a,X_b)$ that are adjacent in $\cal{C}$ such that \\ \qquad \;\;\; $|u(X_a) \setminus \{X_b\}| \geq l$, using order $(V)$;
			\Repeat
			\State Choose a (new) set $S \subseteq u(X_a) \setminus \{X_b\}$ with $|S|=l$, using order$(V)$;
			\If{$X_a \bigCI X_b \,|\, S$\,}
			\State Delete the edge $X_a - X_b$ from $\cal{C}$;
			\State Let {\it sepset\,}$(X_a,X_b)=\,${\it sepset\,}$(X_b,X_a)=S$;
			\EndIf
			\Until{$X_a$ and $X_b$ are no longer adjacent in $\cal{C}$ or all $S \subseteq u(X_a) \setminus \{X_b\}$ with $|S|=l$ have\\ \qquad \;\;\; been considered}
			\Until{all ordered pairs of adjacent vertices $(X_a,X_b)$ in $\cal{C}$ with $|u(X_a) \setminus \{X_b\}| \geq l$ have been\\ \;\;\;\; considered}
			\Until{all pairs of adjacent vertices $(X_a,X_b)$ in $\cal{C}$ satisfy $|u(X_a) \setminus \{X_b\}| \leq l$}
			\State \textbf{Output} : The estimated skeleton $\cal{C}$, separation sets {\it sepset}.
		\end{algorithmic}
	\end{algorithm}

	\begin{algorithm}[!ht]
		\caption{Step 2 of the PC-stable algorithm} \label{PC-stable2}
		
		\begin{algorithmic}%[1]
			\State \textbf{Require} : Skeleton $\cal{C}$, separation sets {\it sepset}.
			\For{all all pair of nonadjacent vertices $X_a, X_c$ with common neighbor $X_b$ in $\cal{C}$}
			\If{$X_b \notin sepset(X_a, X_c)$\,}
			\State Replace $X_a - X_b - X_c$ in $\cal{C}$ by $X_a \rightarrow X_b \leftarrow X_c$;
			\EndIf
			\EndFor
			\State In the resulting PDAG, try to orient as many undirected edges as possible by repeated applications of the following rules :
			\State {\bf (R1)} Orient $X_b - X_c$ into $X_b \rightarrow X_c$ whenever there is an arrow $X_a \rightarrow X_b$ such that $X_a$ and $X_c$ are nonadjacent (otherwise a new v-structure is created).
			\State {\bf (R2)} Orient $X_a - X_c$ into $X_a \rightarrow X_c$ whenever there is a chain $X_a \rightarrow X_b \rightarrow X_c$ (otherwise a directed cycle is created).
			\State {\bf (R3)} Orient $X_a - X_c$ into $X_a \rightarrow X_c$ whenever there are two chains $X_a - X_b \rightarrow X_c$ and $X_a - X_d \rightarrow X_c$ such that $X_b$ and $X_d$ are nonadjacent (otherwise a new v-structure or a directed cycle is created).
		\end{algorithmic}
	\end{algorithm}
	
	% Algorithm \ref{FCI-stable} presents the pseudocode of the FCI-stable algorithm (Section 4.4 in Colombo and Maathuis, 2014).
	
	\begin{algorithm}[!ht]
		\caption{The FCI-stable algorithm (oracle version)} \label{FCI-stable}
		
		\begin{algorithmic}%[1]
			\State \textbf{Require} : Conditional independence information among all variables in $V_X$ given $V_T$.
			\State Use Algorithm \ref{FCI-stable1} to find an initial skeleton ($\cal{C}$), separation sets (sepset) and unshielded triple list ($\cal{M}$);
			\State Use Algorithm \ref{FCI-stable2} to orient v-structures (update $\cal{C}$);
			\State Use Algorithm \ref{FCI-stable3} to find the final skeleton (update $\cal{C}$ and sepset);
			\State Use Algorithm \ref{FCI-stable2} to orient v-structures (update $\cal{C}$);
			\State Use rules (R1)-(R10) of Zhang\,(2008) to orient as many edge marks as possible (update $\cal{C}$);
			\State \textbf{Output} : $\cal{C}$, sepset.
		\end{algorithmic}
	\end{algorithm}

	\begin{algorithm}[!ht]
		\caption{Obtaining an initial skeleton in the FCI-stable algorithm (Algorithm 4.1 in the supplement of Colombo et al., 2012)} \label{FCI-stable1}
		
		\begin{algorithmic}%[1]
			\State \textbf{Require} : Conditional independence information among all variables in $V_X$ given $V_T$, and an ordering order$(V_X)$ on the variables.
			\State Form the complete undirected graph $\cal{C}$ on the vertex set $V_X$ with edges $\circ\hspace{-0.05in}-\hspace{-0.05in}\circ$.
			\State Let $l=-1$;
			\Repeat
			\State $l=l+1$;
			\For{all vertices $X_a$ in $\cal{C}$}
			\State let $u(X_a) = adj(\cal{C}, X_a)$
			\EndFor
			\Repeat
			\State Select a (new) ordered pair of vertices $(X_a,X_b)$ that are adjacent in $\cal{C}$ such that \\ \qquad \;\;\; $|u(X_a) \setminus \{X_b\}| \geq l$, using order $(V_X)$;
			\Repeat
			\State Choose a (new) set $Y \subseteq u(X_a) \setminus \{X_b\}$ with $|Y|=l$, using order$(V_X)$;
			\If{$X_a \bigCI X_b \,|\, Y \cup V_T$\,}
			\State Delete the edge $X_a \circ\hspace{-0.05in}-\hspace{-0.05in}\circ X_b$ from $\cal{C}$;
			\State Let sepset$(X_a,X_b)=\,$sepset$(X_b,X_a)=Y$;
			\EndIf
			\Until{$X_a$ and $X_b$ are no longer adjacent in $\cal{C}$ or all $Y \subseteq u(X_a) \setminus \{X_b\}$ with $|Y|=l$ have\\ \qquad \;\; been considered}
			\Until{all ordered pairs of adjacent vertices $(X_a,X_b)$ in $\cal{C}$ with $|u(X_a) \setminus \{X_b\}| \geq l$ have been\\ \;\;\;\; considered}
			\Until{all pairs of adjacent vertices $(X_a,X_b)$ in $\cal{C}$ satisfy $|u(X_a) \setminus \{X_b\}| \leq l$}
			\State Form a list $\cal{M}$ of all unshielded triples $\langle X_c\,,\cdot\,,X_d \rangle$ (i.e., the middle vertex is left unspecified) in $\cal{C}$ with $c<d$.
			\State \textbf{Output} : $\cal{C}$, sepset, $\cal{M}$.
		\end{algorithmic}
	\end{algorithm}
	
	\begin{algorithm}[!ht]
		\caption{Orienting v-structures in the FCI-stable algorithm (Algorithm 4.2 in the supplement of Colombo et al., 2012)} \label{FCI-stable2}
		
		\begin{algorithmic}%[1]
			\State \textbf{Require} : Initial skeleton ($\cal{C}$), separation sets (sepset) and unshielded triple list ($\cal{M}$).
			\For{all elements $\langle X_a, X_b, X_c \rangle$ of $\cal{M}$}
			\If{$X_b \notin \textrm{sepset}(X_a,X_c)$}
			Orient $X_a \star\hspace{-0.05in}-\hspace{-0.06in}\circ X_b \circ\hspace{-0.06in}-\hspace{-0.05in}\star X_c$ as $X_a \star\hspace{-0.05in}\rightarrow X_b \leftarrow\hspace{-0.05in}\star X_c$
			\EndIf
			\EndFor
			\State \textbf{Output} : $\cal{C}$, sepset.
		\end{algorithmic}
	\end{algorithm}
	
	%\vspace*{-10in}
	
	\begin{algorithm}[!ht]
		\caption{Obtaining the final skeleton in the FCI-stable algorithm (Algorithm 4.3 in the supplement of Colombo et al., 2012)} \label{FCI-stable3}
		
		\begin{algorithmic}%[1]
			\State \textbf{Require} : Partially oriented graph ($\cal{C}$) and separation sets (sepset).
			\For{all vertices $X_a$ in $\cal{C}$}
			\State let $v(X_a) = \mathrm{pds}(\cal{C}, X_a, \cdot)$;
			\For{all vertices $X_b \in adj(\cal{C},X_a)$}
			\State Let $l=-1$;
			\Repeat
			\State $l=l+1$;
			\Repeat
			\State Choose a (new) set $Y \subseteq v(X_a) \setminus \{X_b\}$ with $|Y|=l$;
			\If{$X_a \bigCI X_b \,|\, Y \cup V_T$\,}
			\State Delete the edge $X_a \star\hspace{-0.06in}-\hspace{-0.06in}\star X_b$ from $\cal{C}$;
			\State Let sepset$(X_a,X_b)=\,$sepset$(X_b,X_a)=Y$;
			\EndIf
			\Until{$X_a$ and $X_b$ are no longer adjacent in $\cal{C}$ or all $Y \subseteq v(X_a) \setminus \{X_b\}$ with $|Y|=l$ have\\ \qquad \qquad \; been considered}
			\Until{$X_a$ and $X_b$ are no longer adjacent in $\cal{C}$ or $|v(X_a) \setminus \{X_b\}| < l$}
			\EndFor
			\EndFor
			\State Reorient all edges in $\cal{C}$ as $\circ\hspace{-0.05in}-\hspace{-0.05in}\circ$.
			\State Form a list $\cal{M}$ of all unshielded triples $\langle X_c\,,\cdot\,,X_d \rangle$ in $\cal{C}$ with $c<d$.
			\State \textbf{Output} : $\cal{C}$, sepset, $\cal{M}$.
		\end{algorithmic}
	\end{algorithm}


\clearpage
	
	\vspace{0.2in}
	
	{\small
		
		\begin{thebibliography}{9}
			
			\bibitem{anand}
			\par\noindent\hangindent2.3em\hangafter 1
			\textsc{Anandkumar, A., Tan, V.Y.F., Huang, F.} and \textsc{Willsky, A.S.} (2012). High-Dimensional Gaussian Graphical Model Selection: Walk Summability and Local Separation Criterion. {\it Journal of Machine Learning Research}, {\it 13}, 2293-2337.
			
			\bibitem{ali2009}
			\par\noindent\hangindent2.3em\hangafter 1
			\textsc{Ali, R.A., Richardson, T.S.} and \textsc{Spirtes, P.}
			(2009). Markov equivalence for ancestral graphs. {\it Annals of Statistics}, {\it 37}(5B) 2808-2837.
			
			
			\bibitem{cz}
			\par\noindent\hangindent2.3em\hangafter 1
			\textsc{Chakraborty, S.} and \textsc{Zhang, X.} (2019). Distance Metrics for Measuring Joint Dependence with Application to Causal Inference. {\it Journal of the American Statistical Association}, \textit{114}(528), 1638-1650.
			
			\bibitem{colombo12}
			\par\noindent\hangindent2.3em\hangafter 1
			\textsc{Colombo, D., Maathuis, M.H., Kalisch, M.} and \textsc{Richardson, T.S.}
			(2012). Learning high-dimensional directed acyclic graphs with latent and selection variables. {\it Annals of Statistics}, {\it 40}(1) 294-321.
			
			\bibitem{colombo14}
			\par\noindent\hangindent2.3em\hangafter 1
			\textsc{Colombo, D.} and \textsc{Maathuis, M.H.} (2014). Order-independent constraint-based causal structure learning. {\it Journal of Machine Learning Research}, \textit{15} 3921-3962.
			
			
			\bibitem{hd}
			\par\noindent\hangindent2.3em\hangafter 1
			\textsc{Harris, N.} and \textsc{Drton, M.} (2013). PC Algorithm for Nonparanormal Graphical Models. {\it Journal of Machine Learning Research}, \textit{14} 3365-3383.
			
			
			\bibitem{kb}
			\par\noindent\hangindent2.3em\hangafter 1
			\textsc{Kalisch, M.} and \textsc{B\"uhlmann, P.} (2007). Estimating High-Dimensional Directed Acyclic Graphs with the PC-Algorithm. {\it Journal of Machine Learning Research}, \textit{8} 613-636.
			
			\bibitem{lauritzen}
			\par\noindent\hangindent2.3em\hangafter 1
			\textsc{Lauritzen, S.L.} (1996). Graphical models . {\it Oxford University Press}.
			
			\bibitem{li}
			\par\noindent\hangindent2.3em\hangafter 1
			\textsc{Li, R.}, \textsc{Zhong, W.} and \textsc{Zhu, L.} (2012). Feature selection via distance correlation learning. {\it Journal of the American Statistical Association}, \textit{107}(499)  1129-1139.
			
			\bibitem{liu}
			\par\noindent\hangindent2.3em\hangafter 1
			\textsc{Liu, J.}, \textsc{Li, R.} and \textsc{Wu, R.} (2014). Feature selection for varying coefficient models with ultrahigh-dimensional covariates. {\it Journal of the American Statistical Association}, \textit{109}(505)  266-274.
			
			\bibitem{lb}
			\par\noindent\hangindent2.3em\hangafter 1
			\textsc{Loh, P-L.} and \textsc{B\"uhlmann, P.} (2014). High-Dimensional Learning of Linear Causal Networks via Inverse Covariance Estimation. {\it Journal of Machine Learning Research}, \textit{15} 3065-3105.
			
			\bibitem{handbook}
			\par\noindent\hangindent2.3em\hangafter 1
			\textsc{Maathuis, M., Drton, M., Lauritzen, S.} and \textsc{Wainwright, M}.
			(2019). Handbook of graphical models. {\it CRC Press}.
			
			\bibitem{mb}
			\par\noindent\hangindent2.3em\hangafter 1
			\textsc{Meinshausen, N.} and \textsc{B\"uhlmann, P.} (2010). Stability selection. {\it Journal of the Royal Statistical Society}, \textit{72}(4) 417–473.
			
			
			\bibitem{pearl}
			\par\noindent\hangindent2.3em\hangafter 1
			\textsc{Pearl, J.}
			(2000). Causality. {\it Cambridge University Press}.
			
			\bibitem{r}
			\par\noindent\hangindent2.3em\hangafter 1
			\textsc{Resnick, S. I.} (1999). A Probability Path. {\it Springer}.
			
			\bibitem{rp2002}
			\par\noindent\hangindent2.3em\hangafter 1
			\textsc{Richardson, T.S.} and \textsc{Spirtes, P.}
			(2002). Ancestral graph markov models. {\it Annals of Statistics}, {\it 30}(4) 962-1030.
			
			\bibitem{serfling}
			\par\noindent\hangindent2.3em\hangafter 1
			\textsc{Serfling, R. J.} (1980).
			Approximation Theorems of Mathematical Statistics . {\it Wiley} , New York.
			
			
			\bibitem{ss2013}
			\par\noindent\hangindent2.3em\hangafter 1
			\textsc{Shah, R.D.} and \textsc{Samworth, R.J.} (2013). Variable selection with error control: another look at
			stability selection. {\it Journal of the Royal Statistical Society}, \textit{75}(1) 55–80.
			
			\bibitem{sheng}
			\par\noindent\hangindent2.3em\hangafter 1
			\textsc{Sheng, T.} and \textsc{Sriperumbudur, B.K.} (2019). On distance and kernel measures of conditional independence. arXiv:1912.01103.
			
			\bibitem{ss}
			\par\noindent\hangindent2.3em\hangafter 1
			\textsc{Sondhi, A.} and \textsc{Shojaie, A.} (2019). The Reduced PC-Algorithm: Improved Causal Structure Learning in Large Random Networks. {\it Journal of Machine Learning Research}, {\it 20}, 1-31.
			
			\bibitem{sgs}
			\par\noindent\hangindent2.3em\hangafter 1
			\textsc{Spirtes, P.}, \textsc{Glymour, C.} and \textsc{Scheines, R.} (2000). Causation, Prediction, and Search. The MIT Press, 2nd edition.
			
			\bibitem{spirtes2001}
			\par\noindent\hangindent2.3em\hangafter 1
			\textsc{Spirtes, P.} (2001). An anytime algorithm for causal inference. {\it Proceedings of the $8^{th}$ International Workshop on Artificial Intelligence and Statistics}, 213-221.
			
			\bibitem{sun}
			\par\noindent\hangindent2.3em\hangafter 1
			\textsc{Sun, X., Janzing, D., Sch\"olkopf, B.} and \textsc{Fukumizu, K.} (2007). A kernel-based causal learning algorithm. {\it Proceedings of the $24^{th}$ International Conference on Machine Learning}, 855-862.
			
			\bibitem{srb2007}
			\par\noindent\hangindent2.3em\hangafter 1
			\textsc{Sz\'ekely, G. J., Rizzo, M. L.} and \textsc{Bakirov, N. K}.
			(2007). Measuring and testing independence by correlation of distances. {\it Annals of Statistics}, {\it 35}(6) 2769-2794.
			
			\bibitem{sr2014}
			\par\noindent\hangindent2.3em\hangafter 1
			\textsc{Sz\'ekely, G. J.} and \textsc {Rizzo, M. L.} (2014). Partial
			distance correlation with methods for dissimilarities. {\it Annals of Statistics}, {\it 42}(6) 2382-2412.
			
			\bibitem{tsam}
			\par\noindent\hangindent2.3em\hangafter 1
			\textsc{Tsamardinos, I., Brown, L.E.} and \textsc{Aliferis, C.F.} (2006). The max-min hill-climbing Bayesian network structure learning algorithm. {\it Machine Learning}, \textit{65} 31–78.
			
			\bibitem{uhler}
			\par\noindent\hangindent2.3em\hangafter 1
			\textsc{Uhler, C., Raskutti, G., B\"uhlmann, P.} and \textsc {Yu, B.} (2013). Geometry of the faithfulness assumption in causal inference. {\it Annals of Statistics}, {\it 41}(2) 436-463.
			
			\bibitem{vp1990}
			\par\noindent\hangindent2.3em\hangafter 1
			\textsc{Verma, T.} and \textsc{Pearl, J.}
			(1990). Equivalence and synthesis of causal models. {\it Proceedings of the Sixth Annual Conference on Uncertainty in Artificial Intelligence}, 255-270.
			
			\bibitem{voorman}
			\par\noindent\hangindent2.3em\hangafter 1
			\textsc{Voorman, A., Shojaie, A.} and \textsc{Witten, D.} (2014). Graph estimation with joint additive models. {\it Biometrika}, \textit{99}(1) 1-25.
			
			\bibitem{wang}
			\par\noindent\hangindent2.3em\hangafter 1
			\textsc{Wang, X.}, \textsc{Wenliang, P.}, \textsc{Hu, W.}, \textsc{Tian, Y.} and \textsc{Zhang, H.} (2015). Conditional distance correlation. {\it Journal of the American Statistical Association}, \textit{110}(512) 1726-1734.
			
			\bibitem{wen}
			\par\noindent\hangindent2.3em\hangafter 1
			\textsc{Wen, C.}, \textsc{Wenliang, P.}, \textsc{Huang, M.} and \textsc{Wang, X.} (2018). Sure Independence Screening Adjusted for Confounding Covariates with Ultrahigh Dimensional Data. {\it Statistica Sinica}, \textit{28} 293-317.
			
			\bibitem{zhang2008}
			\par\noindent\hangindent2.3em\hangafter 1
			\textsc{Zhang, J.} (2008). On the completeness of orientation rules for causal discovery in the presence of latent confounders and selection bias. {\it Artificial Intelligence}, \textit{172} 1873–1896.
			
			\bibitem{zhang2018}
			\par\noindent\hangindent2.3em\hangafter 1
			\textsc{Zhang, K., Peters, J., Janzing, D.} and \textsc{Sch\"olkopf, B.} (2012). Kernel-based conditional independence test and application in causal discovery. arXiv:1202.3775.
			
			
		\end{thebibliography}
	}
\end{document}